\documentclass[11pt,a4paper]{article}

\usepackage[centertags]{amsmath}
\usepackage{amsfonts}
\usepackage{amssymb}
\usepackage{amsthm}
\usepackage{epsfig}
\usepackage{setspace}
\usepackage{ae} 
\usepackage{eucal} 
\usepackage[authoryear]{natbib}

\usepackage[usenames]{color}

\theoremstyle{plain}
\newtheorem{thm}{Theorem}[section]
\newtheorem{lem}[thm]{Lemma}

\newtheorem{prop}[thm]{Proposition}
\theoremstyle{definition}

\theoremstyle{remark}
\newtheorem{exmp}[thm]{Example}
\newtheorem{rem}[thm]{Remark}

 \DeclareMathOperator*{\esssup}{esssup}
\DeclareMathOperator*{\essinf}{essinf}

\newcommand{\twert}[2]{\!\!\begin{array}[t]{r}#1\\[-1.00ex]{\scriptstyle(#2)}\end{array}\!\!}

\newcommand{\twertt}[2]{\!\!\begin{array}[t]{c}#1\\[-1.00ex]{\scriptstyle #2 }\end{array}\!\!}

\begin{document}

\title{A primal-dual algorithm for BSDEs}

\author{Christian Bender$^{1}$, Nikolaus Schweizer$^1$, and Jia Zhuo$^2$}

 \maketitle
\footnotetext[1]{Saarland University, Department of Mathematics,
Postfach 151150, D-66041 Saarbr\"ucken, Germany, {\tt
bender@math.uni-sb.de}; {\tt
schweizer@math.uni-sb.de}. \\ Financial support by the Deutsche Forschungsgemeinschaft under grant BE3933/5-1 is gratefully acknowledged.} \footnotetext[2]{       Department of
Mathematics, University of Southern California, 3620 S. Vermont Ave., KAP 104
Los Angeles, CA 90089-2532, {\tt  jiazhuo@usc.edu}.}

\begin{abstract}
We generalize the primal-dual methodology, which is popular in the pricing of early-exercise options, to a 
backward dynamic programming equation associated with time discretization schemes of (reflected) backward stochastic differential equations (BSDEs).
Taking as an input some approximate solution of the backward dynamic program, which was pre-computed, e.g., by least-squares Monte Carlo, this methodology allows to construct 
a confidence interval for the unknown true solution of the time discretized (reflected) BSDE at time $0$. We numerically demonstrate the practical applicability of our method
in two five-dimensional nonlinear pricing problems where tight price bounds were previously unavailable. 
\\[0.2cm] \emph{Keywords:} Backward SDE, numerical approximation, Monte Carlo, option pricing.
\\[0.2cm] \emph{AMS classification:} 65C30, 65C05, 91G20, 91G60.
\end{abstract}

\section{Introduction}

In this paper we aim at constructing tight confidence intervals for the solution $(Y_i)_{i=0,\ldots,n}$ of a dynamic programming equation of the form
\begin{equation}\label{RBSDE}
Y_i=\max\{S_i,E_i[Y_{i+1}] +f(i,Y_i, E_i[\beta_{i+1} Y_{i+1}])\Delta_i\},\quad Y_n=S_n
\end{equation}
at time $i=0$. Dynamic programming equations of the form (\ref{RBSDE}) naturally arise in time discretization schemes for (reflected) BSDEs. We assume that $(\Omega,\mathcal{F},(\mathcal{F}_i)_{i=0,\ldots n},P)$ is a filtered probability space and $E_i[\cdot]$ denotes conditional expectation given $\mathcal{F}_i$. The reflecting 
barrier $S$ is an adapted $\mathbb{R}\cup\{-\infty\}$-valued process, $\beta$ is an adapted $\mathbb{R}^D$-valued process related to the driver of the BSDE 
(e.g. suitably truncated and normalized increments of a $D$-dimensional Brownian motion), the generator $f:\Omega \times \{0,\ldots, n-1\} \times \mathbb{R} \times \mathbb{R}^D \rightarrow \mathbb{R}$
is an adapted random field, and $\Delta_i$ are constants which can be thought of as the stepsizes of a time discretization scheme.  Appropriate integrability and continuity assumptions will be specified later on. 

The special case $f\equiv 0$ of (\ref{RBSDE}) is the well-known recursion for the valuation of Bermudan options. In the wake of the financial crisis, there is an increased interest in
 `small' nonlinearities in pricing. These are due, e.g.,  to counterparty risk or funding risk -- and  had largely been neglected in practice. 
 Building on the BSDE literature and its early pricing applications such as \citet{Berg}, \citet{DuffieEtAl} or the examples in \citet{EKPQ}, the number of pricing problems which have been formulated as BSDEs -- and thus have a discretization of the form (\ref{RBSDE}) -- is steadily growing. 
 Recent examples include funding risk \citep{laurent2012overview, crepey2013counterparty}, counterparty risk \citep{crepey2013counterparty, henry2012counterparty}, model uncertainty \citep{guyon2010uncertain, alanko2013reducing},
 and hedging under transaction costs \citep{guyon2010uncertain}. 
In some of these examples the nonlinearity depends on the delta or the gamma of the option, which can be incorporated in our discrete time setting by choosing the weights $\beta$ 
appropriately.
  The aim of the present paper is to provide a unified and numerically efficient framework for calculating upper and lower price bounds for these problems -- parallel to the well-known primal-dual bounds in Bermudan option pricing.

 The error due to the time discretization (\ref{RBSDE}) for BSDEs driven by a Brownian motion has been 
thoroughly analyzed in the literature under various regularity conditions. We refer to \cite{Zh,BT,GL,GM} for the non-reflected case (corresponding to $S_i \equiv -\infty$ for $i<n$)
and to \cite{BP,MZ,BC} for the reflected case. We emphasize that the results in the present paper can also be applied to the time discretization schemes for BSDEs driven by a Brownian motion with generators with quadratic growth 
as in \cite{CR}, time discretization schemes for BSDEs with jumps considered in \cite{BE08}, and the time discretization scheme for fully nonlinear parabolic PDEs by \cite{FTW}.  

A standard procedure for solving an equation of type (\ref{RBSDE}) numerically is the so-called approximate 
dynamic programming approach. Here, the conditional expectations in (\ref{RBSDE}) are replaced by some approximate conditional expectations operator.
The main difficulty of this approach is, that in each step backwards in time a conditional expectation must be computed numerically, building on the 
approximate solution one step ahead. This leads to a high order nesting of conditional expectations. Hence, it is crucial that the approximate conditional expectations operator 
 can be nested several times without exploding computational cost. Among the techniques which have been applied and analyzed in the
 context of BSDEs driven by a (high-dimensional) Brownian motion are 
least-squares Monte Carlo \citep{LGW,BD}, quantization \citep{BP}, Malliavin Monte Carlo \citep{BT}, cubature on Wiener space \citep{CM}, and sparse grid methods \citep{ZGZ}.

Although convergence rates are available in the literature for these different methods, the quality of the numerical approximation in the practically relevant pre-limit situation is typically 
difficult to assess. Generalizing the primal-dual methodology, which was introduced by \cite{AB} in the context Bermudan option pricing, we suggest to take the numerical solution of the approximate dynamic program as
an input, in order to construct a confidence interval for $Y_0$ via a Monte Carlo approach. In a nutshell, the rationale is to find a maximization problem and a minimization problem with value $Y_0$,
for which optimal controls are available in terms of the true solution $(Y_i)_{i=0,\ldots,n}$ of the dynamic program (\ref{RBSDE}). Using the approximate solution instead of the true one, 
then yields suboptimal controls for these two optimization problems. If the numerical procedure in the approximate dynamic program was successful, these controls are close to optimal and lead to tight lower 
and upper bounds for $Y_0$. Unbiased estimators for the lower and the upper bound can finally be computed by plain Monte Carlo, which results in a confidence interval for $Y_0$. 
\cite{BSUQ} provides a different a posteriori criterion for BSDEs which is better suited for qualitative convergence analysis than for deriving quantitatively meaningful bounds on $Y_0$. The two approaches are thus complimentary.

The paper is organized as follows: In Section 2 we briefly discuss some basic properties of the dynamic programming equation (\ref{RBSDE}) and show how it arises in our two numerical examples, funding risk and counterparty risk. The case of a convex generator $f$ is treated 
in Section 3. In Section 3.1 we first suggest a pathwise approach to the dynamic programming equation (\ref{RBSDE}) which avoids the evaluation of conditional expectations 
in the backward recursion 
in time. This pathwise approach depends on the choice of a $(D+1)$-dimensional martingale and leads to the construction 
of supersolutions for  (\ref{RBSDE}) and to a minimization problem over martingales 
with value process $Y_i$. We then note in Section 3.2
that, due to convexity, $Y_i$ can also be represented as the supremum over a class of classical optimal stopping problems. 
This representation can be thought of 
as a discrete time, reflected analogue of a result in \cite{EKPQ} for continuous time, non-reflected BSDEs driven by a Brownian motion. 
If we think of this maximization problem as the primal problem, then the pathwise approach can be interpreted as a dual minimization problem in the sense 
of information relaxation. This type of duality was first
 introduced independently by \cite{Ro} and \cite{HK} in the context of Bermudan option pricing, and was later extended by \cite{BSS} 
to general discrete time stochastic control problems. Finally, in Section 3.3 we provide some discussion of how the tightness of the bounds depends on the quality of the 
input approximations used in their construction.

In Section 4.1 we explain, how the representations for $Y_0$ as the value of a 
 maximization and a minimization problem can be exploited in order to construct confidence intervals for $Y_0$ via Monte Carlo simulation. This algorithm generalizes the primal-dual algorithm of \cite{AB}
from optimal stopping problems (i.e., the case $f\equiv 0)$   to the case of convex  generators. We also suggest some generic control variates which 
turn out to be powerful in our numerical examples. Numerical examples for the pricing of a European and a Bermudan option on the maximum of five assets under different interest rates for borrowing and lending (funding risk)
are presented in Section 4.2. For constructing the input approximations, we apply the least-squares Monte Carlo algorithm of \cite{LGW} and its martingale basis variant by \cite{BS2} with just a few (up to seven) basis functions. This turns out to be sufficient for achieving very tight 95\% confidence intervals with relative error of typically less than 1\% between lower and 
upper confidence bound in our five-dimensional test examples. 

For non-convex generators $f$, we suggest in Sections 5.1 and 5.2 to apply the input approximation of the approximate dynamic program in order to construct an auxiliary generator $f^{up}$, which is convex and dominates $f$, and
another one $f^{low}$, which is concave and dominated by $f$. This construction can be done in a way that (evaluated at the true solution) $f^{up}$ and $f^{low}$ converge to $f$, when the input approximation of the approximate dynamic program
approaches the true solution. The methods of Section 3 and a corresponding result for the concave case can then be applied to the convex generator $f^{up}$ and to the concave generator $f^{low}$ in order to build 
a confidence interval for $Y_0$ in the general case. In Section 5.3, we test the performance of this algorithm in two applications, the previous example of funding risk and a model of counterparty credit risk where the driver is neither concave nor convex. Again, tight price bounds can be achieved. Appendix A sets our discrete time results into the context of their continuous time analogues.

\section{Discrete time reflected BSDEs}

Suppose $(\Omega,\mathcal{F},(\mathcal{F}_i)_{i=0,\ldots n},P)$ is a filtered probability space in discrete time. We consider the discretized version of a reflected
BSDE of the form (\ref{RBSDE}). Throughout the paper we make the following assumptions:  The time increments $\Delta_i$, $i=0,\ldots, n-1$, are positive real numbers.  $S$ is an adapted process with values in $\mathbb{R}\cup \{-\infty\}$ such that 
$$
\sum_{i=0}^{n-1} E[|S_i{\bf 1}_{\{S_i>-\infty\}}|] + E[|S_n|]<\infty.
$$ 
The random field $f:\Omega \times \{0,\ldots, n-1\} \times \mathbb{R} \times \mathbb{R}^D \rightarrow \mathbb{R}$ is measurable, $f(\cdot,y,z)$ is adapted for every $(y,z)\in \mathbb{R} \times \mathbb{R}^D$,
and $
\sum_{i=1}^{n-1} E[|f(i,0,0)|]<\infty.
$
Moreover, there are adapted, nonnegative processes $\alpha^{(d)}$, $d=0,\ldots, D$ such that the stochastic Lipschitz condition
$$
|f(i,y,z)-f(i,y',z')|\leq \alpha^{(0)}_i |y-y'|+\sum_{d=1}^D  \alpha^{(d)}_i |z_d-z_d'|
$$
holds for every $(y,z), (y',z') \in \mathbb{R} \times \mathbb{R}^D$. Finally, $\beta$ is a bounded, adapted $\mathbb{R}^D$-valued process and the following relations hold:
\begin{eqnarray}
 \alpha^{(0)}_i < \frac{1}{\Delta_i} , \qquad   \sum_{d=1}^D  \alpha^{(d)}_i |\beta_{d,i+1}| \leq \frac{1}{\Delta_i}. \label{Lip2}\label{Lip1}
\end{eqnarray}
A straightforward contraction mapping argument shows that under these assumptions there exists a unique adapted and integrable process $Y$ such that  (\ref{RBSDE}) is satisfied.

\begin{exmp}\label{exmp:intro}
To illustrate the setting let us introduce the two nonlinear pricing problems which also appear in our numerical experiments: Pricing with different interest rates for borrowing and lending, and pricing in a reduced-form model of counterparty credit risk. Going back to \citet{Berg}, the first one is a standard
example in the BSDE literature. \cite{laurent2012overview} have recently emphasized its practical relevance in the context of funding risk. Following the financial crisis there has also been increased interest in credit risk models similar to the second example, see \citet{pallavicini2012funding, crepey2013counterparty, henry2012counterparty} and the references therein.

(i) Let there be a financial market with two riskless and $D$ risky assets. The prices of the risky assets $X_t^1,\ldots X_t^D$ evolve according to
\[
dX_t^d=X_t^d\left(\mu_t^d dt+ \sum_{k=1}^D \sigma_t^{d,k} dW_t^k \right),
\]
where $W$ is a standard $D$-dimensional Brownian motion, and where $\mu$ and $\sigma$ are predictable and bounded processes. Moreover, $\sigma$ is assumed to be a.s. invertible with bounded inverse. The filtration is given by the usual augmented Brownian filtration. The two riskless assets have bounded and predictable short rates $R^l_t$ and $R^b_t$ with $R^l_t \leq R^b_t$ a.s. These are the interest rates for lending and borrowing, i.e., an investor can only hold positive positions in the first one, and only negative ones in the second. 
Consider a square-integrable European claim $h(X_T)$ with maturity $T$. It is well-known that a replicating portfolio for $h(X_T)$ is characterized by two processes $Y_t\in \mathbb{R}$ and $Z_t\in \mathbb{R}^D$ which solve the BSDE
 \begin{equation}\label{RRBSDE}
 dY_t= -f(t, Y_t,Z_t)dt+ Z_t^\top dW_t
 \end{equation}
with terminal condition $Y_T= h(X_T)$ where 
 \[
 f(t,y,z)= -R_t^l y - z^\top \sigma_t^{-1}\left(\mu_t-R_t^l \bar{1} \right)+ (R_t^b-R_t^l)\left(y-z^\top \sigma_t^{-1}\bar{1}\right)_-,   
 \]
see \citet{Berg} or the survey paper of \citet{EKPQ}.
Here, $\bar{1}$ denotes the vector containing only ones in $\mathbb{R}^D$ and $^\top$ denotes matrix transposition. The function $f$ is both convex and  Lipschitz continuous. Our key quantity of interest is the claim's fair price at time $0$ given by  $Y_0$.  $Z^\top_t \sigma_t^{-1}$ corresponds to the vector of amounts of money invested in the risky assets at time $t$.
 Discretizing time and taking conditional expectations gives a recursion of the form (\ref{RBSDE}) for $Y$, see  \cite{Zh,BT}. In this European case $S_i \equiv -\infty$ for $i<n$ and $S_n\equiv h(X_T)$.  The discretization of the $Z$-part is given by $Z_i=E_i\left[\frac{W_{t_{i+1}}-W_{t_i}}{t_{i+1}-t_i } Y_{i+1}\right]$. This corresponds to a vector of Malliavin derivatives of $Y$ and is thus naturally related to a delta hedge.
 In view of (\ref{RBSDE}) we would thus like to choose $\beta_{i+1}=(W_{t_{i+1}}-W_{t_i})/(t_{i+1}-t_i)$. However, in order to fulfill condition (\ref{Lip2}) we have to truncate the Brownian increments at some value. 
 Since $(W_{t_{i+1}}-W_{t_i})/(t_{i+1}-t_i)$ is a vector of normal random variables with standard deviation of order $(t_{i+1}-t_i)^{-\frac12}$, 
 we can make this truncation error small as the time discretization gets finer, taking into account Lipschitz continuity and the factor $\Delta_i=t_{i+1}-t_i$ outside $f$, see e.g. \cite{LGW}. 
 For a Bermudan or American claim, (\ref{RRBSDE}) is replaced by a suitable reflected BSDE. In the discretization, $S_i$ is then the payoff from exercising at time $t_i$.
 
(ii) The second example is a special case of the model of counterparty credit risk due to \cite{DuffieEtAl}. We change the setting of (i) by assuming there is only one riskless asset with rate $R_t$ which can be both borrowed and lent. Moreover, we consider a risk-neutral valuation framework, i.e., $\mu_t=R_t \bar{1}$.
Given a square-integrable European claim $h(X_T)$ with maturity $T$, we denote by $Y_t$ the claim's fair price at time $t$ conditional on no default having occurred yet. The claim's possible default is modelled through a stopping time which is the first jump time of a Poisson process with intensity $Q_t$. 
Here, $Q_t=Q(Y_t)$ is a decreasing, continuous and bounded function of  $Y_t$, i.e., if the claim's value is low, default becomes more likely. 
If default occurs at time $t$, the claim's holder receives a fraction $\delta \in [0,1)$ of the current value $Y_t$. Following Proposition 3 in \cite{DuffieEtAl}, the value process is then characterized by the nonlinear relation
\[
Y_t=E_t \left[\int_t^T f(s,Y_s)ds+h(X_T)\right],
\]
where $f(t,y)= -(1-\delta)Q(y)y-R_t y$. 
Discretizing naturally leads to an equation of type (\ref{RBSDE}) with $\beta\equiv 0$.  Condition (\ref{Lip2}) then reduces to the requirement that the time discretization is sufficiently fine. 
\end{exmp}

The dynamic programming equation (\ref{RBSDE}) implies that the solution $Y$ also solves the optimal stopping problem
$$
Y_i=\esssup_{\tau \in \mathcal{S}_{i}} E_i\left[S_\tau+\sum_{j=i}^{\tau-1} f(j,Y_j, E_j[\beta_{j+1} Y_{j+1}])\Delta_j\right],\quad i=0,\ldots, n,
$$
where $\mathcal{S}_i$ is the set of stopping times with values bigger or equal to $i$. This optimal stopping problem is unusual in the sense that the reward upon stopping depends 
on the Snell envelope $Y_j$. Note that one can pose restrictions on the set of admissible stopping times by 
choosing the set $\{(i,\omega);\; S_i(\omega)=-\infty\}$, at which exercise is never optimal. We can hence restrict the supremum in this optimal stopping problem 
to the subset $\bar{\mathcal{S}}_{i}\subset \mathcal{S}_{i}$ of stopping times $\tau$ which take values in $\mathcal{E}(\omega)=\{i;\; S_i(\omega)>-\infty\}$.
 An optimal stopping time is given by
\begin{equation}\label{tau*}
\tau^*_i=\inf\{j\geq i;\; S_j\geq E_j[Y_{j+1}]+f(j,Y_j, E_j[\beta_{j+1}Y_{j+1}])\Delta_j  \}\wedge n
\end{equation}
We also note the following alternative representation of $Y_i$ via optimal stopping of a nonlinear functional.
\begin{prop}\label{prop:stopping}
 For every $i=0,\ldots,n$, 
$$
Y_i=\esssup_{\tau \in \bar{\mathcal{S}}_{i}} Y^{(\tau)}_i
$$
where $(Y^{(\tau)}_j)_{j\geq i}$ solves the dynamic programming equation
$$
Y^{(\tau)}_j=E_j[Y^{(\tau)}_{j+1}] +f(j,Y^{(\tau)}_j, E_j[\beta_{j+1} Y^{(\tau)}_{j+1}])\Delta_j, \; i\leq j < \tau,\quad \ Y^{(\tau)}_\tau=S_\tau
$$
Moreover, the stopping time $\tau^*_i$, defined in (\ref{tau*}) is optimal.
\end{prop}

This representation is a direct consequence of the following simple, but useful, comparison theorem. For nonreflected discrete time BSDEs  related 
comparison results can be found in \citet{CE} and \citet{CS} under different sets of assumptions.
\begin{prop}\label{prop:comparison}
 Suppose there are stopping times $\sigma\leq \tau $ such that for every $\sigma\leq  i < \tau$ 
\begin{eqnarray*}
 Y_i^{up}&\geq&\max\{S_i,E_i[Y^{up}_{i+1}] +f(i,Y^{up}_i, E_i[\beta_{i+1} Y^{up}_{i+1}])\Delta_i\} \\ 
 Y_i^{low}&\leq&\max\{S_i,E_i[Y^{low}_{i+1}] +f(i,Y^{low}_i, E_i[\beta_{i+1} Y^{low}_{i+1}])\Delta_i\}
\end{eqnarray*}
and $Y_\tau^{up}\geq Y_\tau^{low}$.
Then, under the standing assumptions, $Y_i^{low}\leq Y^{up}_i$ holds for every  $\sigma\leq  i \leq \tau$.
\end{prop}
\begin{proof}
We define, for $i=1,\ldots,n$,
$$
\Delta Y_i=(Y_i^{up}-Y_i^{low}){\bf 1}_{\{\sigma \leq i \leq \tau\}}.
$$
It is sufficient to show that $\Delta Y_i\geq 0$ for every $i=1,\ldots,n$. 
 We prove this assertion by backward induction and note that it holds in the case $i=n$ by assumption. Now, 
suppose that $\Delta Y_{i+1}\geq 0$ is already shown. Then, 
on the set $\{Y_i^{low}>S_i\}\cap \{\sigma\leq i<\tau\}$ we obtain by the Lipschitz assumption on $f$,
\begin{eqnarray*}
\Delta Y_i &\geq &  E_i[\Delta Y_{i+1}] +(f(i,Y^{up}_i, E_i[\beta_{i+1} Y^{up}_{i+1}])-f(i,Y^{low}_i, E_i[\beta_{i+1} Y^{low}_{i+1}]))\Delta_i \\
&\geq &E_i\left[\Delta Y_{i+1}\left(1- \sum_{d=1}^D  \alpha^{(d)}_i |\beta_{d,i+1}| \Delta_i\right)\right] - \alpha_i^{(0)} |\Delta Y_i| \Delta_i\\
&\geq & -\alpha_i^{(0)} |\Delta Y_i| \Delta_i,
\end{eqnarray*}
which yields $\Delta Y_i\geq 0$. On the set $\{Y_i^{low}\leq S_i\}\cup \{i\geq\tau\}\cup \{i<\sigma\}$, the inequality $\Delta Y_i\geq 0$ is obvious.
\end{proof}

\section{The case of a convex generator}\label{SecConvex}

In Sections 3.1 and 3.2,  we discuss how to construct `tight' supersolutions and subsolutions to the dynamic programming equation (\ref{RBSDE})  when the generator $f$ is convex in $(y,z) \in \mathbb{R}^{1+D}$.
These constructions are based, respectively, on the choice of suitable martingales and control processes. In Section 3.3 we present, for the special case where $f$ does not depend on $z$,
error estimates which quantify how these choices affect the quality of the resulting error bounds.

\subsection{Upper bounds}
 We first consider a pathwise approach
which leads to supersolutions of the dynamic program due to the convexity of $f$. Roughly speaking, the idea is to remove all conditional expectations from equation (\ref{RBSDE}) and subtract martingale increments,
wherever conditional expectations were removed. To this end, let us fix a one-dimensional martingale $M^0$ and an $\mathbb{R}^D$-valued martingale $M$ such that
\[
\sum_{d=1}^D \sum_{i=0}^{n-1} E[\alpha^{(d)}_i |M_{d,i+1}-M_{d,i}|] \Delta_i<\infty.
\]
The set of all such pairs $(M^0,M)$ is denoted by $\mathcal{M}_{1+D}$.
Given $(M^0,M)\in \mathcal{M}_{1+D}$ we define the non-adapted process
$\theta^{up}_i=\theta^{up}_i(M^0,M)$ via
 \begin{equation}\label{pathwise}
\theta^{up}_i=\max\{S_i,\theta^{up}_{i+1}-(M^0_{i+1}-M^0_i) +f(i,\theta^{up}_i, \beta_{i+1} \theta^{up}_{i+1}-(M_{i+1}-M_i) )\Delta_i\},\quad \theta^{up}_n=S_n.
\end{equation}
Once the martingales are chosen, this recursion can be solved path by path. The stochastic Lipschitz condition on $f$ and the assumptions on $M$ ensure that 
$\theta^{up}_i$ is integrable. Hence, after 
solving the recursion,  we can take conditional expectation once, instead of taking \textit{nested} conditional expectations as in the original dynamic program (\ref{RBSDE}). Exploiting the convexity of $f$ 
we shall now show that $E_i[\theta^{up}_i]$ is always an upper bound for $Y_i$, and that $Y_i$ can be recovered by a suitable choice of the martingales. We recall that the 
martingale part $N$ of the Doob decomposition of an integrable stochastic process $V$ (Doob martingale of $V$, for short) is given by
$$
N_i:=\sum_{j=0}^{i-1} (V_{j+1}-E_j[V_{j+1}]),\quad i=0,\ldots,n.
$$
\begin{thm}\label{thm:convexdual}
 Suppose $f$ is convex in $(y,z)$. Then, for every $i=0,\ldots,n$,
\begin{eqnarray*}
 Y_i=\essinf_{(M^0,M)\in \mathcal{M}_{1+D}} E_i[\theta^{up}_i(M^0,M)]
\end{eqnarray*}
where $\theta^{up}(M^0,M)$ is defined by the pathwise dynamic programming equation (\ref{pathwise}). Moreover, 
the martingale $( M^{0,*}, M^*)$, where $ M^{0,*}$ and $ M^*$ are the Doob martingales of $Y$ and $\beta Y$, is optimal even in the sense of pathwise control, i.e.
$$
\theta_i^{up}(  M^{0,*}, M^*) = Y_i,\quad P\textnormal{-a.s.}
$$  
\end{thm}
\begin{proof}
By the convexity of $f$
and of the max-operator as well as the martingale property,
 we obtain,
$$
E_i[\theta^{up}_i] \geq \max\{S_i, E_i[E_{i+1}[\theta^{up}_{i+1}]]+f(i,E_i[\theta^{up}_i], E_i[\beta_{i+1} E_{i+1}[\theta^{up}_{i+1}]]) \Delta_i\}
$$
Consequently, $E_i[\theta^{up}_i(M^0,M)]$ is a supersolution of (\ref{RBSDE}) and by the comparison result of Proposition \ref{prop:comparison}
$$
E_i[\theta^{up}_i(M^0,M)]\geq Y_i.
$$
We now choose ${M}^{0,*}$ and $ M^*$ as the Doob martingales of $Y$ and $\beta Y$, respectively, and note that $(M^{0,*},M^*)\in \mathcal{M}_{1+D}$, because, thanks to (\ref{Lip2}),
\begin{eqnarray*}
 \sum_{d=1}^D \sum_{i=0}^{n-1} E[\alpha^{(d)}_i |M_{d,i+1}-M_{d,i}|] \Delta_i&=&\sum_{d=1}^D \sum_{i=0}^{n-1} E[\alpha^{(d)}_i |\beta_{d,i+1}Y_{i+1}-E_i[\beta_{d,i+1}Y_{i+1}]|] \Delta_i\\
&\leq& 2 \sum_{i=0}^{n-1} E[|Y_{i+1}|], 
\end{eqnarray*}
which is finite. We 
claim that
$$
 \theta^{up, *}_i:=\theta^{up}_i( M^{0,*},  M^*) = Y_i
$$ 
almost surely. This will be shown by backward induction on $i$, with the case $i=n$ being trivial. 
Suppose that the claim is true for $i+1$. Then, making use of the definition of the Doob martingale,
\begin{eqnarray*}
 \theta^{up, *}_i&=&\max\{S_i,Y_{i+1}-(Y_{i+1}-E_i[Y_{i+1}]) +f(i,\theta^{up, *}_i, \beta_{i+1} Y_{i+1}-(\beta_{i+1} Y_{i+1}-E_i[\beta_{i+1} Y_{i+1}]) )\Delta_i\} \\
&=& \max\{S_i, E_i[Y_{i+1}] +f(i,\theta^{up, *}_i, E_i[\beta_{i+1} Y_{i+1}]) \Delta_i\}.
\end{eqnarray*}
By the Lipschitz property of $f$ in the $y$-variable, a straightforward contraction mapping argument shows that $\theta^{up, *}_i=Y_i$, which finishes the proof.
\end{proof}

The previous theorem can be applied to compute upper confidence bounds on $Y_0$. To this end one first chooses a $(1+D)$-dimensional martingale, which one thinks is close to the Doob martingale 
of $(Y,\beta Y)$. This can e.g. be (related to) the Doob martingale of an approximation $\tilde Y$ of $Y$ which was pre-computed by an algorithm of one's choice. Then one solves the pathwise 
dynamic program in (\ref{pathwise}) and finally approximates the expectation by averaging over sample paths. The details of such an implementation are discussed in Section \ref{sec:algorithm} below. One issue, which arises
in this approach, is that the pathwise dynamic program is not explicit in time, as $\theta^{up}_i$ appears on both sides of the equation. It can be solved by a Picard iteration  to a given 
precision. In some situations, the following explicit expression in terms of a pathwise maximization problem is advantageous.
\begin{prop}\label{prop:max}
 Suppose $f$ is convex in $(y,z)$, and define the convex conjugate in the $y$-variable by
$$
f^{\#y}(\omega,i,r,z)=\sup\{ry -f(\omega,i,y,z);\;\quad y \in \mathbb{R}\}
$$
which is defined on
$$
D^{(i,\omega,z)}_{f^{\#y}}:=\{ r \in \mathbb{R};\; f^{\#y}(\omega,i,r,z)<\infty\}\subset [-\alpha^{(0)}_i(\omega),\alpha^{(0)}_i(\omega)].
$$ 
Then, for $(M^0,M) \in \mathcal{M}_{1+D}$ and $i=0,\ldots,n-1$, $\theta_i^{up}$ as defined in (\ref{pathwise}) can be rewritten as
\begin{eqnarray}\label{eq:explicit}
 \theta^{up}_i=\max\left \{S_i, \sup_{r \in D^{(i,\omega,z^{up}_i(\omega))}_{f^{\#y}}} \frac{1}{1-r\Delta_i}\left(\theta^{up}_{i+1}-(M^0_{i+1}-M^0_i) - f^{\#y}(i,r, z^{up}_i )\Delta_i\right)\right\},
\end{eqnarray}
where $z^{up}_i=\beta_{i+1} \theta^{up}_{i+1}-(M_{i+1}-M_i)$.
\end{prop}
\begin{proof}
 By convexity, we have $f(i,\cdot)=(f(i,\cdot)^{\#y})^{\#r}$, where $^{\#r}$ denotes the convex conjugate in the $r$-variable of $f^{\#y}(i,\cdot)$. Hence,
$$
\theta^{up}_i=\max\left\{S_i, \sup_{r \in D^{(i,\omega,z^{up}_i(\omega))}_{f^{\#y}}} \left(\theta^{up}_{i+1}-(M^0_{i+1}-M^0_i) +r\theta^{up}_i\Delta_i- f^{\#y}(i,r, z^{up}_i )\Delta_i\right)\right\}
$$
By a similar argument as on p. 36 in \citet{EKPQ} the supremum is achieved at some $r^*\in D^{(i,\omega,z^{up}_i(\omega))}_{f^{\#y}}$. 
Indeed, notice first that the extension of $f^{\#y}$ to an $\mathbb{R}\cup\{+\infty\}$-valued function on the real line via $f^{\#y}(\omega,i,r,z)=+\infty$ 
for $r\notin D^{(i,\omega,z)}_{f^{\#y}}$ is lower-semicontinuous
in $r$ by Theorem 12.2 and p. 52 in \citet{Rock}. By the boundedness of the set $D^{(i,\omega,z^{up}_i(\omega))}_{f^{\#y}}$, there is a sequence $r_k:=r_k(\omega)$  converging to a limit $r^*=r^*(\omega)$ 
in the closure of $D^{(i,\omega,z^{up}_i(\omega))}_{f^{\#y}}$ such that 
\begin{eqnarray*}
&& \sup_{r \in D^{(i,\omega,z^{up}_i(\omega))}_{f^{\#y}}} \left(\theta^{up}_{i+1}-(M^0_{i+1}-M^0_i) +r\theta^{up}_i\Delta_i- f^{\#y}(i,r, z^{up}_i )\Delta_i\right)\\
&=& \lim_{k\rightarrow \infty}  \left(\theta^{up}_{i+1}-(M^0_{i+1}-M^0_i) +r_k\theta^{up}_i\Delta_i- f^{\#y}(i,r_k, z^{up}_i )\Delta_i\right) \\
&\leq & \left(\theta^{up}_{i+1}-(M^0_{i+1}-M^0_i) +r^*\theta^{up}_i\Delta_i- f^{\#y}(i,r^*, z^{up}_i )\Delta_i\right),
\end{eqnarray*}
where the  inequality is due to the lower-semicontinuity. This implies 
$$f^{\#y}(i,r^*, z^{up}_i )<\infty.$$ 
Thus, $r^*(\omega)\in D^{(i,\omega,z^{up}_i(\omega))}_{f^{\#y}}$
and it  attains the supremum. Hence,
$$
\theta^{up}_i=\max\{S_i,  \theta_{i+1}^{up}-(M^0_{i+1}-M^0_i) + r^*\theta^{up}_i\Delta_i- f^{\#y}(i,  r^*, z^{up}_i )\Delta_i\}.
$$
For $\theta^{up}_i>S_i$ we, thus obtain 
$$
\theta^{up}_i= \frac{1}{1-  r^* \Delta_i}\left(\theta^{up}_{i+1}-(M^0_{i+1}-M^0_i) - f^{\#y}(i,r^*, z^{up}_i )\Delta_i \right). 
$$
Consequently, $\theta^{up}_i$ is dominated by the right hand side of the assertion. The reverse inequality can be shown in the same way.
\end{proof}

\begin{exmp}\label{exmp:thetafunding}
 In Example \ref{exmp:intro}  (i), 
$$
f^{\#y}(i,r,z)=z^\top \sigma_i^{-1}(\mu_i+r\bar{1})
$$
and the maximizer must belong to  the set $\{-R_i^b, -R_i^l\}$, because $f(i,\cdot)=(f(i,\cdot)^{\#y})^{\#r}$. Hence, for the European option case, a recursion for $\theta^{up}$, which is explicit in time, reads
$$
\theta^{up}_i= \sup_{r \in \{-R_i^b, -R_i^l\}} \frac{1}{1-r\Delta_i}\Bigl(\theta^{up}_{i+1}-(M^0_{i+1}-M^0_i) -  [\beta_{i+1} \theta^{up}_{i+1}-(M_{i+1}-M_i) )]^\top\sigma_i^{-1}(\mu_i+r\bar{1} ) \Delta_i\Bigl).
$$
\end{exmp}

\subsection{Lower bounds}
We now turn to the construction of subsolutions.
In order to derive a maximization problem with value process given by $Y_i$, we
denote by $f^\#$ the convex conjugate of $f$ in $(y,z)$, i.e.
$$
f^\#(\omega,i,r,\rho)=\sup\{ry+ \rho^\top z -f(\omega,i,y,z);\;\quad (y,z) \in \mathbb{R}^{1+D}\}
$$
which is defined on
$$
D^{(i,\omega)}_{f^\#}:=\{(r,\rho)\in \mathbb{R}^{1+D};\; f^\#(\omega,i,r,\rho)<\infty\}\subset \prod_{d=0}^D [-\alpha^{(d)}_i(\omega),\alpha^{(d)}_i(\omega)].
$$
We also define
$$
\mathcal{U}_i(f^\#):=\{(r_j,\rho_j)_{j\geq i} \textnormal{ adapted};\; \sum_{j=i}^{n-1} E[|f^\#(j,r_j,\rho_j)|]\Delta_j<\infty\},
$$
and note that $(f^\#(j,r_j,\rho_j))_{j\geq i}$ is an adapted process for $(r,\rho)\in \mathcal{U}_i(f^\#)$. The next lemma in particular shows that $\mathcal{U}_i(f^\#)$ is nonempty.
\begin{lem}\label{lem:subgradient}
 Suppose $f$ is convex in $(y,z)$, $(\tilde Y)_{j\geq i}$ is an $\mathbb{R}$-valued adapted and integrable process, and $(\tilde Z_j)_{j\geq i}$ is an $\mathbb{R}^D$-valued adapted process such that
$$
\sum_{d=1}^D \sum_{j=i}^{n-1}  E[\alpha_j^{(d)} |\tilde Z_{d,j}|]\Delta_j<\infty.
$$
Then there is a pair $(\tilde r,\tilde \rho)\in \mathcal{U}_i(f^\#)$ such that for $j=i,\ldots,n-1$,
$$
\tilde r_j\tilde Y_j+{\tilde \rho_j}^\top \tilde Z_j-f^\#(j,\tilde r_j,\tilde \rho_j)=f(j,\tilde Y_j, \tilde Z_j).
$$
\end{lem}
\begin{proof}
Similarly to the proof of Proposition 6.1 in \citet{CS}, we exploit the existence of a measurable subgradient due to Theorem 7.10 in \citet{CKV}. The latter theorem guarantees
for every $j=i,\ldots, n-1$ existence 
of an $\mathcal{F}_j$-measurable random vector $(\tilde r_j,\tilde \rho_j)$ such that
$$
f(j, \tilde Y_j+y,\tilde Z_j+z)- f(j,\tilde Y_j,\tilde Z_j)\geq \tilde r_j y+ \tilde \rho_j^\top z
$$ 
for every $(y,z)\in \mathbb{R}^{1+D}$. Taking the supremum over $(y,z)\in \mathbb{R}^{1+D}$, one has
$$
\tilde r_j\tilde Y_j+{\tilde \rho_j}^\top \tilde Z_j-f^\#(j,\tilde r_j,\tilde \rho_j)\geq f(j,\tilde Y_j, \tilde Z_j).
$$
In particular, $(\tilde r_j(\omega),\tilde \rho_j(\omega))\in D^{(j,\omega)}_{f^\#}$. The converse inequality immediately follows from $f^{\#\#}=f$ by convexity. So it remains to show that
$$
\sum_{j=i}^{n-1} E[|f^\#(j,\tilde r_j,\tilde \rho_j)|]\Delta_j<\infty.
$$
By the stochastic Lipschitz property of $f$ and the boundedness of $D^{(j,\omega)}_{f^\#}$ we obtain
\begin{eqnarray*}
&& \sum_{j=i}^{n-1} E[|f^\#(j,\tilde r_j,\tilde \rho_j)|]\Delta_j=  \sum_{j=i}^{n-1} E[|\tilde r_j\tilde Y_j+{\tilde \rho_j}^\top \tilde Z_j-f(j,\tilde Y_j, \tilde Z_j)|]\Delta_j \\
&\leq & 2 \sum_{j=i}^{n-1} E[\alpha_j^{(0)} \Delta_j |\tilde Y_j|] + 2 \sum_{d=1}^D \sum_{j=i}^{n-1} E[\alpha_j^{(d)} |\tilde Z_{d,j}|]\Delta_j +  \sum_{j=i}^{n-1} E[|f(j,0,0)|]\Delta_j<\infty
\end{eqnarray*}
thanks to (\ref{Lip2}).
\end{proof}

The following result is a discrete time reflected 
analogue of Proposition 3.4 in \citet{EKPQ}. For discrete time (non-reflected) BSDEs a similar result (for convex generators in $z$ only) can be found in \citet{CS} under a different set of assumptions.
\begin{thm}\label{thm:primal}
Suppose $f$ is convex in $(y,z)$. Let 
$$
\theta_i^{low}(\tau,r,\rho) :=\Gamma_{i,\tau}(r,\rho)S_\tau -
\sum_{j=i}^{\tau-1} \Gamma_{i,j}(r,\rho) \frac{f^\#(j,r_j,\rho_j)\Delta_j}{1-r_j \Delta_j} \;\text{ where }\; \Gamma_{i,j}(r,\rho):= \prod_{k=i}^{j-1} \frac{1+\rho_k^\top  \beta_{k+1} \Delta_k}{1-r_k \Delta_k}.
$$
Then,
\begin{eqnarray*}
 Y_i=\esssup_{\tau \in \bar{\mathcal{S}}_i}  \esssup_{(r,\rho) \in \mathcal{U}_i(f^\#)} E_i[\theta_i^{low}(\tau,r,\rho)]
\end{eqnarray*}
Maximizers exist and are given by any $(r_j^*,\rho_j^*)_{j\geq i}$ such that for $j=i,\ldots,n-1$,
\begin{equation}\label{convexopt}
r^*_jY_j+{\rho^*_j}^\top E_j[\beta_{j+1}Y_{j+1}]-f^\#(j,r^*_j,\rho^*_j)=f(j,Y_j, E_j[\beta_{j+1}Y_{j+1}])
\end{equation}
and $\tau_i^*$ as defined in (\ref{tau*}).
\end{thm}
\begin{proof}
 Fix $i\in\{0,\ldots, n\}$. Given a stopping time $\tau\in \bar{\mathcal{S}}_i$ and a pair $(r,\rho) \in \mathcal{U}_i(f^\#)$ define
$$
Y_j(\tau, r,\rho):= E_j\left[\theta_j^{low}(\tau,r,\rho) \right], \quad i\leq j \leq \tau.
$$
Then, $Y_\tau(\tau, r,\rho)=S_\tau$ and, for $i\leq j < \tau$, 
\begin{eqnarray*}
 Y_j(\tau, r,\rho)&= & E_j[Y_{j+1}(\tau, r,\rho)]+(\rho_j^\top E_j[\beta_{j+1}Y_{j+1}(\tau, r,\rho) ]+ r_j Y_j(\tau, r,\rho)- f^\#(j,r_j,\rho_j))\Delta_j 
\\ &\leq&  E_j[Y_{j+1}(\tau, r,\rho)] + f(j,Y_j(\tau, r,\rho), E_j[\beta_{j+1}Y_{j+1}(\tau, r,\rho) ])\Delta_j,
\end{eqnarray*}
where the last estimate is due to the fact that $f^{\#\#}=f$ by convexity. Now the comparison result in Proposition \ref{prop:comparison} and Proposition \ref{prop:stopping} 
imply
$$
Y_i(\tau, r,\rho) \leq Y^{(\tau)}_i\leq Y_i.
$$ 
For the converse inequality, we first notice that the Lemma \ref{lem:subgradient} yields existence of  a pair of processes 
$(r_j^*,\rho_j^*)_{j\geq i} \in \mathcal{U}_i(f^\#)$ such that (\ref{convexopt}) holds, because by (\ref{Lip2})
\begin{eqnarray*}
 \sum_{d=1}^D \sum_{j=i}^{n-1}  E[\alpha_j^{(d)} |E_j[\beta_{d,j+1}Y_{j+1}]|]\Delta_j \leq  \sum_{j=i}^{n-1} E\left[ |Y_{j+1}|\right]<\infty.
\end{eqnarray*}
 Then, by the definition of $\tau^*_i$, we obtain for $i\leq j <\tau^*_i$
\begin{eqnarray*}
 Y_j&=&E_j[Y_{j+1}] +f(j,Y_j, E_j[\beta_{j+1} Y_{j+1}])\Delta_j \\ &=& E_j[Y_{j+1}]+ \left(r^*_j Y_j+{\rho^*_j}^\top E_j[\beta_{j+1}Y_{j+1}]-f^\#(j,r^*,\rho^*_j)\right)\Delta_j.
\end{eqnarray*}
As $Y_{\tau^*_i}=S_{\tau^*_i}$, we conclude that
$$
Y_i=Y_i(\tau^*_i, r^*,\rho^*)
$$ 
by uniqueness for this dynamic programming equation.
\end{proof}
\begin{rem}
 If we think of the representation in Theorem \ref{thm:primal} as a `primal' maximization problem, then the representation in Theorem \ref{thm:convexdual} can be interpreted
as a dual minimization problem in the sense of information relaxation. This dual approach was introduced for Bermudan option pricing by \citet{Ro} and \citet{HK}, and was further developed 
for discrete time stochastic control problems by \citet{BSS}. Indeed, given a martingale $(M^0,M)\in \mathcal{M}_{1+D}$, we define 
\begin{eqnarray*}
\frak p_{M^0,M}: \{i,\ldots,n\} \times  \prod_{j=i}^{n-1} D^{(j,\omega)}_{f^\#} &\rightarrow & L^1(\Omega,P), \\
(k,(r,\rho))&\mapsto&\sum_{j=i}^{k-1} \Gamma_{i,j}(r,\rho)\frac{(M^0_{j+1}-M^0_j)+ \rho_j^\top  ( M_{j+1}-M_j)\Delta_j}{1-r_j\Delta_j}.
\end{eqnarray*}
Then, for every $(\tau, (r,\rho))\in \bar{\mathcal{S}}_i\times  \mathcal{U}_i(f^\#)$,
\begin{equation}\label{eq:penalty}
E_i[\frak p_{M^0,M}(\tau,r,\rho)]=0.
\end{equation}
We next relax the adaptedness property of the controls $(\tau, (r,\rho))$ and observe that, by Theorem \ref{thm:primal} and (\ref{eq:penalty}),
\begin{eqnarray*}
 Y_i&=& \esssup_{\tau \in \bar{\mathcal{S}_i}}  \esssup_{(r,\rho) \in \mathcal{U}_i(f^\#)} E_i[\theta_i^{low}(\tau,r,\rho) - \frak p_{M^0,M}(\tau,r,\rho)] \\
&\leq &  E_i\Bigg[\max_{k=i,\ldots, n}\; \max_{ \twertt{ \scriptstyle (r_j,\rho_j)\in  D^{(j,\omega)}_{f^\#}}{ \scriptscriptstyle j=i,\ldots,k-1\; }}\; \left(\theta_i^{low}(k,r,\rho) - \frak p_{M^0,M}(k,r,\rho) \right) \Bigg]\\
&=:& E_i[\tilde \theta_i(M^0,M)].
\end{eqnarray*}
Notice that the maximum on the right hand side of the inequality is taken pathwise, which means that we may now choose anticipating controls. The rationale of the information relaxation approach is that 
one allows for anticipating controls, but subtracts a penalty, here $\frak p_{M^0,M}$. The penalty does not penalize non-anticipating controls by (\ref{eq:penalty}). We say that a penalty 
$\frak p^*$ is optimal, if it
penalizes anticipating controls in a way that the pathwise maximum is achieved at a non-anticipating control. This  implies
$$
Y_i=E_i \Bigg[\max_{k=i,\ldots, n}\; \max_{ \twertt{ \scriptstyle (r_j,\rho_j)\in  D^{(j,\omega)}_{f^\#}}{ \scriptscriptstyle j=i,\ldots,k-1\; }}\; \left(\theta_i^{low}(k,r,\rho)      - \frak p^*(k,r,\rho)\right) \Bigg].
 $$
In the present setting, one can show that
$$
\tilde \theta_i(M^0,M)=\theta_i^{up}(M^0,M).
$$
To see this, one first derives a recursion formula for $\tilde \theta_i(M^0,M)$ and then follows the arguments behind Proposition 
\ref{prop:max}. In particular, Theorem \ref{thm:convexdual} shows that an optimal penalty is given by $\frak p_{ M^{0,*}, M^*}$.
\end{rem}

\subsection{Error estimates}

We now provide some error analysis of the lower and upper bounds for the convex case: How does the accuracy of the input approximations affect the tightness of the upper and lower bounds?
For simplicity, we focus here on the case where $\beta \equiv 0$. While obviously restrictive, this case does cover many applications of practical interest, such as  the BSDEs arising
in the credit risk literature, see \citet{crepey2013counterparty, henry2012counterparty}. For the lower bound, we assume that the suboptimal controls are derived from an input approximation
of the process $Y$ in exactly the same way, in which we choose these controls in the algorithm presented in Section 4, cf. \eqref{eq:r} and (\ref{defcontr}).
An in-depth-analysis of the general case would certainly require to pose additional assumption on $\beta$ and is beyond the scope of this paper. 

\begin{thm}
Suppose $f$ is convex in $(y,z)$ and denote by $M^{0,*}$    the Doob martingale of $Y$. \\  
(i) For every $M^0\in \mathcal{M}_1$ and $i=0,\ldots,n-1$,
$$
E_i[ \theta^{up}_i(M^0)]-Y_i\leq  E_i\left[C(i,\alpha^{(0)}) \max_{j=i,\ldots,n} 
|M^0_j-M^{0,*}_j|\right],
$$
where
$$
C(\alpha^{(0)},i)=1+ \prod_{l=i}^{n-1} 
(1-\alpha_l^{(0)}\Delta_l)^{-1}\left(1+\sum_{j=i}^{n} 
\alpha_j^{(0)}\Delta_j\right).
$$
(ii) Let $i=0,\ldots, n-1$. Suppose $(\tilde Q_j)_{j=i,\ldots,n-1}$ is an adapted and integrable approximation of 
$(E_j[Y_{j+1}])_{j=i,\ldots,n-1}$ and $(\tilde Y_j)_{j=i,\ldots,n-1}$ is an adapted and integrable  
approximation of $(Y_j)_{j=i,\ldots,n-1}$. Define
an adapted process $r$ via
\begin{equation}\label{eq:r}
r_j\tilde Y_j-f^\#(j,r_j)=f(j,\tilde Y_j),\quad j=i,\ldots,n-1,
\end{equation}
  and $\tau:=\inf\{j\geq i;\; S_j\geq \tilde Q_j+f(j,\tilde 
Y_j)\Delta_j\}\wedge n$.
Then,
\begin{eqnarray*}
   Y_i-E_i[\theta_i^{low}(\tau,r)]   &\leq &
E_i\left[c(i,\alpha^{(0)},\tau) \left(3\sum_{j=i}^{\tau\wedge (n-1)} |\tilde 
Y_j-Y_j|\alpha^{(0)}_j\Delta_j\right.\right. \\ &&\left.\left.+ \sum_{j=i}^{\tau-1}{\bf 1}_{A_j} (\tilde 
Q_j-E_j[Y_{j+1}])_+ +{\bf 1}_{{A^c_\tau}\cap \{\tau<n\}} 
(E_\tau[Y_{\tau+1}]-\tilde Q_\tau)_+\right) \right]
\end{eqnarray*}
where
\begin{eqnarray*}
  c(i,\alpha^{(0)},\tau)&=& \prod_{l=i}^{\tau \wedge (n-1)} 
(1-\alpha_l^{(0)}\Delta_l)^{-1}, \\
A_j&=&\{S_j\geq E_j[Y_{j+1}]+f(j,Y_j 
)\Delta_j\},\quad j=i,\ldots, n-1.
\end{eqnarray*}
\end{thm}

\begin{rem}
  In the nonreflected case, i.e. $S_i=-\infty$ for $i<n$, we have 
$\tau=n$ and $A_j=\emptyset$ for $j<n$. Hence the lower bound estimate 
simplifies to
$$
 Y_i - E_i[\theta_i^{low}(\tau,r)] \leq 3\,E_i\left[c(i,\alpha^{(0)},n) 
\sum_{j=i}^{n-1} |\tilde Y_j-Y_j|\alpha^{(0)}_j\Delta_j \right].
$$
In the reflected case, the indicators ${\bf 1}_{A_j}$ and ${\bf 1}_{{A^c_\tau}\cap \{\tau<n\}}$ correspond to wrong stopping decisions 
of the approximate stopping time $\tau$ compared to the optimal stopping time. In practice, such wrong stopping decisions rarely occur, when a good 
approximation $\hat Q_j$ of the continuation value  $E_j[Y_{j+1}]$ is applied. Hence, the corresponding terms are not expected to grow linearly in the number of exercise dates, 
although this is suggested by the worst case estimates. For a rigorous statement 
of this intuition in the case of optimal stopping we refer \citet{Be}. 
\end{rem}
 
\begin{proof}
(i) Given a martingale $M^0$ define for $k=0,\ldots, n$ and 
$i=0,\ldots,k-1$
$$
\theta_i^{up,k}(M^0)=\theta_{i+1}^{up,k}(M^0)+f\left(i,\max_{i\leq \kappa 
\leq n} \theta^{up,\kappa}_i(M^0)\right)\Delta_i-(M^0_{i+1}-M^0_i),\quad  
\theta_k^{up,k}(M^0)=S_k,
$$
with the convention that $\theta^{up,k}_i=0$ for $i>k$. Backward 
induction combined with a contraction mapping argument shows that there 
exist a unique solution $\theta_i^{up,k}(M^0)$ such that $$E\left[\max_{i\leq 
k\leq n}|\theta_i^{up,k}(M^0)|\right]<\infty.$$
We claim that $\theta^{up}_i(M^0)$ defined via (\ref{pathwise}) 
coincides with  $\max_{i\leq k \leq n} \theta^{up,k}_i(M^0)$. Indeed,
\begin{eqnarray*}
  \max_{i\leq k \leq n} \theta^{up,k}_i(M^0)&=&\max\left\{S_i, \max_{i+1\leq 
k\leq n}  \theta_{i+1}^{up,k}(M^0)+ f\left(i,\max_{i\leq k \leq n} 
\theta^{up,k}_i(M^0)\right)\Delta_i-(M^0_{i+1}-M^0_i)\right\},\\
\max_{n\leq k \leq n} \theta^{up,k}_n(M^0)&=&S_n.
\end{eqnarray*}
  As this equation has a unique solution, we conclude that
$$
\max_{i\leq k \leq n} \theta^{up,k}_i(M^0)=\theta^{up}_i(M^0).
$$
In particular, thanks to Theorem \ref{thm:convexdual},
\begin{eqnarray}\label{eq:error1}
E_i[ \theta^{up}_i(M^0)]&\leq&E_i\left[\max_{i\leq k \leq n} 
\theta^{up,k}_i(M^{0,*})\right]+ E_i\left[\max_{i\leq k \leq n} 
|\theta^{up,k}_i(M^0)-\theta^{up,k}_i(M^{0,*})|\right]\nonumber\\
&=& Y_i+ E_i\left[\max_{i\leq k \leq n} 
|\theta^{up,k}_i(M^0)-\theta^{up,k}_i(M^{0,*})|\right].
\end{eqnarray}
So it remains to estimate the last term on the right-hand side of 
(\ref{eq:error1}). We denote
$$
\Delta\theta_i=\max_{i\leq k\leq n} 
|\theta^{up,k}_i(M^0)-\theta^{up,k}_i(M^{0,*})-M^{0}_i+M^{0,*}_i|.
$$
Then,
\begin{eqnarray*}
  \Delta\theta_i&\leq &\Delta\theta_{i+1}+\left|f\left(i,\max_{i\leq k \leq n} 
\theta^{up,k}_i(M^0)-M^0_i+M^{0,*}_i\right)\Delta_i-f\left(i,\max_{i\leq k \leq n} 
\theta^{up,k}_i(M^{0,*})\right)\Delta_i\right| \\ &&+
  \left|f\left(i,\max_{i\leq k \leq n} 
\theta^{up,k}_i(M^0)-M^0_i+M^{0,*}_i\right)\Delta_i+f\left(i,\max_{i\leq k \leq n} 
\theta^{up,k}_i(M^{0})\right)\Delta_i\right|.
\end{eqnarray*}
Hence,
$$
\Delta\theta_i\leq  (1-\alpha_i^{(0)}\Delta_i)^{-1}\left( 
\Delta\theta_{i+1}+\alpha_i^{(0)}\Delta_i |M^0_i-M^{0,*}_i|\right),
$$
which in turn implies
$$
\Delta\theta_i\leq \prod_{l=i}^{n-1} 
(1-\alpha_l^{(0)}\Delta_l)^{-1}\left(\Delta\theta_n+\sum_{j=i}^{n} 
\alpha_j^{(0)}\Delta_j |M^0_j-M^{0,*}_j|\right).
$$
Thus,
$$
\max_{i\leq k \leq n} 
|\theta^{up,k}_i(M^0)-\theta^{up,k}_i(M^{0,*})|\leq C(\alpha^{(0)},i) 
\max_{j=i,\ldots,n} |M^0_j-M^{0,*}_j|.
$$
Combining this estimate with (\ref{eq:error1}) finishes the error 
estimate for the upper bound.
\\[0.1cm]
(ii) We now turn to the estimate for the lower bound.
Note first that, by Lemma \ref{lem:subgradient}, there is an adapted 
process $r$ such that (\ref{eq:r}) holds. As (\ref{eq:r}) implies that 
$r_j(\omega)\in D^{(j,\omega)}_{f^\#}$,
  we observe that $|r_j|\leq \alpha_j^{(0)}$.
Define $Y^{low}_j=E_j[\theta^{low}_j(\tau,r)]$ for $i\leq j \leq \tau$. 
Then, as in the proof of Theorem \ref{thm:primal} and making use of the 
relation between $r$ and $\tilde Y$, we obtain
$$
Y^{low}_j=E_j[Y^{low}_{j+1}]+ r_j(Y^{low}_j-\tilde 
Y_j)\Delta_j+f(j,\tilde Y_j)\Delta_j,\quad i\leq j<\tau,\quad 
Y^{low}_\tau=S_\tau.
$$
We now recall that
$$
Y_j=E_j[Y_{j+1}]+ f(j,Y_j)\Delta_j + 
(S_j-E_j[Y_{j+1}]-f(j,Y_j)\Delta_j)_+.
$$
Hence, for $i\leq j<\tau$
\begin{eqnarray*}
  E_i[Y_j-Y^{low}_j]&\leq& E_i[Y_{j+1}-Y^{low}_{j+1}]+2E_i[\alpha^{(0)}_j 
|\tilde Y_j-Y_j|]\Delta_j +E_i[\alpha^{(0)}_j|Y_j-Y^{low}_j|]\Delta_j\\ &&+ 
E_i[(S_j-E_j[Y_{j+1}]-f(j,Y_j)\Delta_j)_+]
\end{eqnarray*}
As $S_j<\tilde Q_j+f(j,\tilde Y_j)\Delta_j$ for $j<\tau$, we obtain
\begin{eqnarray*}
  && Y_i-Y^{low}_i\nonumber \\ &\leq& E_i\left[\prod_{l=i}^{\tau-1} 
(1-\alpha_l^{(0)}\Delta_l)^{-1}\left(Y_\tau-Y^{low}_\tau+ 3 
\sum_{j=i}^{\tau-1} \alpha^{(0)}_j |\tilde Y_j-Y_j|\Delta_j + 
\sum_{j=i}^{\tau-1}{\bf 1}_{A_j} (\tilde 
Q_j-E_j[Y_{j+1}])_+\right)\right].
\end{eqnarray*}
To finish the proof it now suffices to observe that by the definition of $\tau$
\begin{eqnarray*}
  Y_\tau-Y^{low}_\tau&=&{\bf 1}_{{A^c_\tau}\cap \{\tau<n\}}\left( 
E_\tau[Y_{\tau+1}]+f(\tau,Y_\tau)\Delta_\tau-S_\tau \right)\\
&\leq & {\bf 1}_{{A^c_\tau}\cap \{\tau<n\}} 
\left((E_\tau[Y_{\tau+1}]-\tilde Q_\tau)_+ + \alpha^{(0)}_\tau 
|Y_\tau-\tilde Y_\tau| \Delta_\tau\right).
\end{eqnarray*}
\end{proof}

\section{A primal-dual algorithm for the convex case}

\subsection{The algorithm}\label{sec:algorithm}
In this section we explain, how the results of Section \ref{SecConvex} can be applied in order to construct an upper biased estimator, a lower biased estimator, and confidence intervals for $Y_0$ 
in the spirit of the \citet{AB} algorithm for Bermudan option pricing, when $f$ is convex in $(y,z)$. 

\paragraph*{Markovian setting and input approximations.}
We suppose that 
we are in a Markovian setting, i.e. $f(i,\cdot)=F(i,X_i,\cdot)$ and $S_i=G_i(X_i)$ depend on $\omega$ only through an $\mathbb{R}^N$-valued Markovian process $X_i$ where the mappings $F$ and $G$ are measurable 
in the $x$-component and such that the resulting $f$ and $S$ fulfill the conditions postulated in Section 2.
Moreover, $\beta_{i+1}$ is assumed to be independent of $\mathcal{F}_i$. Then, there are deterministic functions $y_i(x),q_i(x),z_{d,i}(x), d=1,\ldots,D$, such that
$$
Y_i=y_i(X_i),\quad E_i[Y_{i+1}]=q_i(X_i),\quad E_i[\beta_{d,i+1} Y_{i+1}]=z_{d,i}(X_i)
$$
and
\begin{eqnarray}\label{eq:intyqz1}
 E\left[ \sum_{i=1}^{n-1} \left(|y_i(X_i)|+|q_i(X_i)| + \sum_{d=1}^D (\alpha^{(d)}_i+ 1) |z_{d,i}(X_i)|\right)\right]<\infty.
\end{eqnarray}
We assume that measurable approximations $\tilde y_i(x), \tilde q_i(x)$ and $\tilde z_{d,i}(x)$ for these functions are pre-computed by some numerical algorithm, such that 
the integrability condition (\ref{eq:intyqz1}) also holds for the tilded expressions. This ensures that the samples in the numerical algorithm below are always drawn from integrable random variables. 
In our numerical experiments a least-squares Monte Carlo 
estimator for the conditional expectations in (\ref{RBSDE}) is applied in order to construct these approximations, but other choices are possible.

\paragraph*{Upper biased estimator.}
Given the approximations $\tilde y_i(x), \tilde q_i(x)$ and $\tilde z_{d,i}(x)$, we sample $\Lambda^{out}$ independent copies  $$(X_i(\lambda),\beta_i(\lambda);\;i=0,\ldots,n)_{\lambda=1,\ldots, \Lambda^{out}}$$ of 
$(X_i,\beta_i;\;i=0,\ldots,n)$, to which we refer as `outer' paths. For the upper confidence bound we apply Theorem \ref{thm:convexdual}. We thus wish to calculate 
$\theta^{up}_i(M^0,M)$ for some martingales $M^0$, $M$, which are `close' to the unknown Doob martingales of $Y$ and $\beta Y$. We apply instead the Doob martingales of the approximations 
$\tilde y(X)$ and $\beta   \tilde y(X)$ to $Y$ and $\beta Y$. Along the $\lambda$th outer path this leads in view of (\ref{pathwise}) to $ \theta^{up}_n(\lambda)=G_n(X_n(\lambda))$
and, for $i=n-1,\ldots,0$,
\begin{eqnarray}\label{closedform}
&& \theta^{up}_i(\lambda)\nonumber \\&=&\max\Bigl\{G_i(X_i(\lambda)),\theta^{up}_{i+1}(\lambda)-(\tilde y_{i+1}(X_{i+1}(\lambda))-E[\tilde y_{i+1}(X_{i+1})|X_i=X_i(\lambda)])\nonumber \\ &&+  F\Bigl(i,X_i(\lambda),\theta^{up}_i(\lambda), \beta_{i+1}(\lambda) \theta^{up}_{i+1}(\lambda)
\nonumber \\ &&\quad-(\beta_{i+1}(\lambda) \tilde y_{i+1}(X_{i+1}(\lambda))-E[\beta_{i+1} \tilde y_{i+1}(X_{i+1})|X_i=X_i(\lambda)] )\Bigr)\Delta_i\Bigr\}.
\end{eqnarray}
Then, by Theorem \ref{thm:convexdual}, the estimator 
$$
\hat Y^{up}:=\frac{1}{\Lambda^{out}} \sum_{\lambda=1}^{\Lambda^{out}}  \theta^{up}_0(\lambda) 
$$
for $Y_0$, which is obtained by averaging over the outer paths,
has a positive bias.
In general, we cannot expect that the conditional expectations in (\ref{closedform}) can be calculated in closed form. 
 Instead we apply a conditionally unbiased estimator for these conditional expectations 
 by averaging over a set of `inner' samples. For each $i$ and each outer path $X(\lambda)$ generate $\Lambda^{in}$ independent copies of $(X_{i+1},\beta_{i+1})$ under the conditional 
law given that $X_i=X_i(\lambda)$. These samples are denoted by $(X_{i+1}(\lambda,l),\beta_{i+1}(\lambda,l)),\; l=1,\ldots, \Lambda^{in}$. We then define the plain Monte Carlo estimators for the 
conditional expectations in (\ref{closedform}) along the $\lambda$th outer paths by
\begin{eqnarray}
\hat E[\tilde y_{i+1}(X_{i+1})|X_i=X_i(\lambda)]&=&\frac{1}{\Lambda^{in}} \sum_{l=1}^{\Lambda^{in}} \tilde y_{i+1}(X_{i+1}(\lambda,l))  \nonumber \\ \hat E[\beta_{i+1} \tilde y_{i+1}( X_{i+1})|X_i=X_i(\lambda)]&=&
\frac{1}{\Lambda^{in}} \sum_{l=1}^{\Lambda^{in}} \beta_{i+1}(\lambda,l) \tilde y_{i+1}( X_{i+1}(\lambda,l)).  \label{inner}
\end{eqnarray}
Then, in the recursive construction for $\theta^{up}_i(\lambda)$ we replace the conditional expectations in (\ref{closedform}) by
the plain Monte Carlo estimators (\ref{inner}) in all instances and apply the notation  $\theta_i^{up,AB}(\lambda)$. The corresponding 
upper bound estimator for $Y_0$ is obtained by averaging over the outer paths
$$
\hat Y^{up,AB}:=\frac{1}{\Lambda^{out}} \sum_{\lambda=1}^{\Lambda^{out}}  \theta_0^{up,AB}(\lambda) .
$$ 
Here, the superscript `$AB$' stands for Andersen and Broadie, who suggested this method for Bermudan options in 2004.
By a straightforward application of Jensen's inequality we observe that, by convexity of the max-operator and of $f$, $\hat Y^{up,AB}$ has an additional positive bias 
compared to $\hat Y^{up}$, which is due to the inner simulations. In particular, $\hat Y^{up,AB}$  has a positive bias as an estimator for $Y_0$. 

 For Bermudan option pricing problems various other 
constructions for the input martingales have been introduced in the literature, see e.g. \citet{BBS}, \citet{DFM} and \citet{SZH}. These constructions can also be adapted to the present BSDE setting.

\paragraph*{Lower biased estimator.}
In order to construct an estimator for $Y_0$ with a negative bias, we define a stopping time
$\tilde \tau(\lambda)$ along the outer paths (i.e. for $\lambda=1,\ldots, \Lambda^{out}$) by 
\begin{eqnarray*}\label{tildetau} \nonumber
\tilde \tau(\lambda)=\inf\{j\geq 0;\; G_j( X_j(\lambda))\geq \tilde q_j( X_j(\lambda))+F(j,X_j(\lambda),\tilde y_j( X_j(\lambda)),\tilde z_j( X_j(\lambda)))\Delta_j  \} \wedge n
\end{eqnarray*}
 and controls $(\tilde r_j(\lambda), \tilde \rho_j(\lambda))_{j=0,\ldots,n-1}\in \mathcal{U}_0(F^\#)$ as (approximate) solutions of 
\begin{eqnarray}\label{defcontr}
 && \tilde r_j(\lambda)\tilde y_j( X_j(\lambda))+\tilde \rho_j(\lambda)^\top \tilde z_j( X_j(\lambda))-F^\#(j,X_{j}(\lambda),\tilde r_j(\lambda),\tilde \rho_j(\lambda)) \nonumber \\ &=&
 F(j,X_{j}(\lambda),\tilde y_j( X_j(\lambda)), \tilde z_j( X_j(\lambda))),
\end{eqnarray}
cf. Lemma \ref{lem:subgradient}.
Then, by Theorem \ref{thm:primal}, the plain Monte Carlo estimator 
\begin{eqnarray*}
\hat Y^{low,AB} &=&\frac{1}{\Lambda^{out}} \sum_{\lambda=1}^{\Lambda^{out}}  \theta_0^{low,AB}(\lambda), \\ 
 \theta_0^{low,AB}(\lambda) &=& 
\Gamma_{0,\tilde \tau(\lambda)}(\tilde r(\lambda),\tilde \rho(\lambda)) G(\tilde \tau(\lambda), X_{\tilde \tau(\lambda)}(\lambda))\\&& +
 \sum_{j=0}^{\tilde \tau(\lambda)-1} \Gamma_{0,j}(\tilde r(\lambda),\tilde \rho(\lambda)) \frac{F^\#(j,X_{j}(\lambda),\tilde r_j(\lambda),\tilde \rho_j(\lambda))\Delta_j}{1-\tilde r_j(\lambda) \Delta_j} 
\end{eqnarray*}
for $Y_0$ has a negative bias.

\paragraph*{Confidence intervals.} 

Starting from the estimator with a positive bias and the one with a negative bias, one can construct asymptotic confidence intervals for $Y_0$ under additional square integrability conditions which ensure 
that
$$
E[|\theta_0^{low,AB}(\lambda)|^2+|\theta_0^{up,AB}(\lambda)|^2]<\infty.
$$
In order to guarantee this, we impose that
\begin{eqnarray*}
 E\left[|G(n,X_n)|^2+ \sum_{i=1}^{n-1} \left(| F(i,X_i,0,0)|^2+|G(i,X_i)|^2 {\bf 1}_{\{G(i,X_i)>-\infty\}}\right)\right]<\infty.
\end{eqnarray*}
This assumption implies that
\begin{eqnarray}\label{eq:intyqz2}
 E\left[ \sum_{i=1}^{n-1} \left(|y_i(X_i)|^2+|q_i(X_i)|^2 + \sum_{d=1}^D (\alpha^{(d)}_i+1)^2 |z_{d,i}(X_i)|^2\right)\right]<\infty
\end{eqnarray}
holds instead of  condition (\ref{eq:intyqz1}). Hence we shall also impose the stronger integrability assumption (\ref{eq:intyqz2}) on the pre-computed approximations
$\tilde y_i(x), \tilde q_i(x)$ and $\tilde z_{d,i}(x)$. This additional assumption  ensures that
$(\tilde r_j(\lambda), \tilde \rho_j(\lambda))_{j=0,\ldots,n-1}$ defined via (\ref{defcontr}) now satisfy
$$
\sum_{j=0}^{n-1} E[|F^\#(j, X_j(\lambda), \tilde r_j(\lambda),\tilde \rho_j(\lambda))|^2]\Delta_j<\infty,
$$
and this square integrability  additionally needs to be assumed, if (\ref{defcontr}) only holds approximately.

Now, with square integrable and independent copies $(\theta_0^{low,AB}(\lambda), \theta_0^{up,AB}(\lambda))$, $\lambda=1,\ldots,\Lambda^{out}$ at hand,
an (asymptotic) 95\% confidence interval $I^{(95)}$ for $Y_0$ can be constructed by adding (resp. subtracting) 1.96 empirical standard deviations 
to the upper estimator (from the lower estimator), i.e.,
\begin{eqnarray}
I^{(95)}= && \left[ 
\hat Y^{low,AB}- 1.96 \left( \frac{1}{\Lambda^{out}(\Lambda^{out}-1)} \sum_{\lambda=1}^{\Lambda^{out}}  (\theta_0^{low,AB}(\lambda) - \hat Y^{low,AB})^2  \right)^\frac12
, \right.\nonumber\\
&& \quad\left.\hat Y^{up,AB}+ 1.96 \left( \frac{1}{\Lambda^{out}(\Lambda^{out}-1)} \sum_{\lambda=1}^{\Lambda^{out}}  (\theta_0^{up,AB}(\lambda) - \hat Y^{up,AB})^2  \right)^\frac12
\right]\nonumber.
\end{eqnarray}
This asymptotic confidence interval is valid even if one applies the same outer paths for the lower estimator which were already used for the upper estimator. Indeed, abbreviating 
$I^{(95)}=[a^{\Lambda^{out}},b^{\Lambda^{out}}]$, one has
\begin{eqnarray*}
 P(\{Y_0\notin I^{(95)}\})&\leq& P(\{Y_0<a^{\Lambda^{out}}\})+ P(\{Y_0>b^{\Lambda^{out}}\})\\ &\leq & P(\{E[\hat Y^{low,AB}]<a^{\Lambda^{out}}\})+ P(\{E[\hat Y^{up,AB}]>b^{\Lambda^{out}}\})
\rightarrow 0.95,
\end{eqnarray*}
as $\Lambda^{out}$ tends to infinity, where we first applied the biasedness of the two estimators and then the central limit theorem separately to both terms.

\paragraph*{Control variates.}
The numerical experiments below (cf. Figure 1) illustrate that the additional bias of the upper bound estimator due to the inner simulations may be substantial with a moderate number of inner paths (say 1,000). It therefore appears to 
be essential to apply variance reduction techniques for the estimation of the conditional expectations in (\ref{closedform}) by Monte Carlo. We suggest some control variates, for which  we merely require 
that 
$$
E[\beta_{d,i+1}],\quad E[\beta_{d,i+1} \,\beta_{d',i+1}],\quad d,d'=1,\ldots,D
$$  
are available in closed form. This is e.g. the case when $\beta_{d,i+1}$ is (up to a constant) given by  truncated increments of independent Brownian motions. In this case we perform an orthogonal 
projection of $\tilde y_{i+1}(X_{i+1})$ on the span of the random variables $(\beta_{1,i+1},\ldots, \beta_{D,i+1})$ under the conditional probability given $X_i$. This orthogonal projection is given by 
\begin{eqnarray*}
 \beta_{i+1}^\top B_{i+1}^{+}E[\beta_{i+1} \tilde y_{i+1}(X_{i+1})|X_i=x],
\end{eqnarray*}
where $B_{i+1}^{+}$ is the Moore-Penrose pseudoinverse of the matrix 
$$
B_{i+1}=(E[\beta_{d,i+1} \,\beta_{d',i+1}])_{d,d'=1,\ldots,D}.
$$
Here, we made use of the assumption that $\beta_{i+1}$ is independent of $\mathcal{F}_i$.
If $\tilde y$ and $\tilde z$ are good approximations of $y$ and $z$, then $\tilde z_i( X_i)$ is also expected to be a good approximation of $E[\beta_{i+1} \tilde y_{i+1}( X_{i+1})|X_i]$. These considerations 
motivate us to replace the estimators (\ref{inner}) for the conditional expectations in (\ref{closedform}) by 
\begin{eqnarray}
&& \hat E^{C}[\tilde y_{i+1}( X_{i+1})|X_i=X_i(\lambda)]  \nonumber\\ 
 &=&E[\beta_{i+1}]^{\top} B_{i+1}^+ \tilde z_i(X_i(\lambda))  
+\frac{1}{\Lambda^{in}} \sum_{l=1}^{\Lambda^{in}} \left(\tilde y_{i+1}( X_{i+1}(\lambda,l))-\beta_{i+1}(\lambda,l)^{\top} B_{i+1}^+ \tilde z_i(X_i(\lambda)) \right) \nonumber 
\end{eqnarray}
and
\begin{eqnarray}
&&\hat E^C[\beta_{i+1} \tilde y_{i+1}( X_{i+1})|X_i=X_i(\lambda)]  = E[\beta_{i+1}]\tilde q_i( X_i(\lambda) )+B_{i+1}B^+_{i+1}\tilde z_i( X_i(\lambda))  \nonumber \\ &&\quad +
\frac{1}{\Lambda^{in}} \sum_{l=1}^{\Lambda^{in}} \beta_{i+1}(\lambda,l) (\tilde y_{i+1}( X_{i+1}(\lambda,l))-\tilde q_i(  X_i(\lambda)) 
-\beta_{i+1}(\lambda,l)^{\top} B_{i+1}^+ \tilde z_i( X_i(\lambda))) ,\label{controlvariate}
\end{eqnarray}
which are still conditionally unbiased. The estimator $\hat Y^{up,ABC}$ is  then calculated 
analogously to $\hat Y^{up,AB}$, but applying (\ref{controlvariate}) instead of (\ref{inner}). Again, by Jensen's inequality,
the `up'-estimator has a positive bias. For the classical optimal stopping problem, a similar control variate for inner simulations was suggested by \citet{BBS} 
in the special case when $\beta_{i+1}$ are increments of independent Brownian motions.

 We also recommend to run the lower bound estimator $\hat Y^{low,AB}$ with a control variate in order to reduce the number of samples $\Lambda^{out}$. In this regard, we suggest the use of
\begin{eqnarray}\label{CVLB}
&& \hspace{-1cm}\sum_{j=0}^{\tilde \tau(\lambda)-1}   \Gamma_{0,j}(\tilde r(\lambda),\tilde \rho(\lambda)) \frac{\tilde y_{j+1}( X_{j+1}(\lambda))- E[\tilde y_{j+1}(X_{j+1})|X_j=X_j(\lambda)]}{1-\tilde r_j(\lambda)\Delta_j}\nonumber \\
&& \hspace{-0.5cm}+\Gamma_{0,j}(\tilde r(\lambda),\tilde \rho(\lambda)) \frac{\tilde \rho_j(\lambda)^\top \Delta_j (\beta_{j+1}(\lambda)\tilde y_{j+1}( X_{j+1}(\lambda))-E[\beta_{j+1}\tilde y_{j+1}( X_{j+1})|X_j=X_j(\lambda)])}{1-\tilde r_j(\lambda)\Delta_j},\;
\end{eqnarray}
if the conditional expectations are available in closed form. If not, a set of `inner' simulations will be required for the construction of the upper bound estimator anyway, and this inner sample 
can be used to estimate the conditional expectations in the control variate (\ref{CVLB}) via (\ref{controlvariate}). The resulting estimator with a negative bias is denoted 
$\hat Y^{low,ABC}$.  The construction of  asymptotic confidence intervals is, of course, completely analogous to the situation without control variates.

\subsection{Numerical examples}\label{sec:num1}

We apply the above algorithm in the context of adjusting the option price value due to funding constraints in the context of Example \ref{exmp:intro} (i). We consider the pricing problem 
of a European and a Bermudan  call spread option with maturity $T$ on the maximum of $D$ assets, which are modeled by independent, identically distributed geometric Brownian motions with drift $\mu$ and volatility $\sigma$ whose values at time $t_i=Ti/n$, $i=0,\ldots,n$ are denoted by $X_{d,i}$. The interest rates $R^b$ and $R^l$ are constant over time.  
The generator $f$ is then given by
\[
F(i,x,y,z)= -R^l y-\frac{\mu-R^l}{\sigma}\sum_{d=1}^D z_d +(R^b-R^l)\left(y-\frac{1}\sigma\sum_{d=1}^D z_d\right)_{-}
\]
We define $\beta_{d,i+1}(t_{i+1}-t_i)$ as the truncated Brownian increment driving the $d$th stock over the period $[t_i,t_{i+1}]$.
The payoff of the option is given by
$$
G_i(x)= \left\{ \begin{array}{cl}\left(\max_{d=1,\ldots,D} x_d-K_1\right)_+-2\left(\max_{d=1,\ldots,D} x_d-K_2\right)_+, & i \in \mathcal{E} \\ -\infty, & i\notin \mathcal{E}.  \end{array}\right.
$$
for  strikes $K_1, K_2$ and a set of time points  $\mathcal{E}$ at which the option can be exercised. Hence, $\mathcal{E}=\{n\}$ gives a European option.
For the Bermudan option case we consider the situation of four exercise dates which are equidistant over the time horizon, i.e. $\mathcal{E}=\{n/4, n/2, 3n/4, n\}$.
Unless otherwise noted, we use the following parameter values: 
$$D=5, \; T=0.25, \; R^l=0.01, \; R^b=0.06, \; X_{d,0}=100, \; \mu=0.05, \; \sigma=0.2, \;K_1=95, \; K_2=115.$$
We first generate approximations $\tilde y^{LGW}, \tilde q^{LGW}, \tilde z^{LGW}$ by the least-squares Monte Carlo algorithm of \citet{LGW}. This algorithm requires the choice of a set 
of basis functions. Then an empirical regression on the span of these basis functions is performed with a set of $\Lambda^{reg}$ sample paths, which are independent of the outer and inner samples required for the 
primal-dual algorithm later on. In the European option case we apply the following sets of basis functions: For the implementation with $b^y=2$ basis functions we choose $1$ and 
$E[G_n( X_n)|X_i=x]$ for the computation of  $\tilde y_i^{LGW}( x), \tilde q_i^{LGW}( x)$ and  $x_d \frac{d}{dx_d} E[G_n( X_n)|X_i=x]$ for the computation of $\tilde z^{LGW}_{d,i}( x)$, $d=1,\ldots,D$.
For call options on the maximum of $D$ Black-Scholes stocks,  closed form expressions for the option price and its delta in terms of a multivariate normal distribution are derived 
in \citet{Jo}. In the present setting, this formula can be simplified to an expectation of a function of a one-dimensional standard normal random variable, see e.g. \citet{BBS}.
As a trade-off between computational time and accuracy,  we approximate this expectation via quantization of the one-dimensional standard normal distribution with  21 grid points.
In the implementation with $b^y=7$ basis functions we additionally apply $x_1,\ldots, x_5$ as basis functions for $\tilde y_i^{LGW}( x), \tilde q_i^{LGW}( x)$, and $x_d$ as a basis function for 
$\tilde z^{LGW}_{d,i}( x)$. For the Bermudan option case we use six basis functions for  $\tilde y_i^{LGW}( x), \tilde q_i^{LGW}( x)$, namely $1$, 
$E[G_j( X_j)|X_i=x]$, $j\in \mathcal{E}$, and $\max_{j\in \mathcal{E},\;j\geq i} E[G_j( X_j)|X_i=x]$. The corresponding deltas $x_d \frac{d}{dx_d} E[G_j( X_j)|X_i=x]$, $j\in \mathcal{E}, j \geq i$, are 
chosen as basis functions for $\tilde z^{LGW}_{d,i}( x)$.

In the European option case, this choice of basis functions also allows to apply the martingale basis algorithm of \citet{BS2}, although a slight bias in the input approximations
is introduced due to the approximation of the basis functions by the quantization approach.  Compared to the generic least-squares Monte Carlo algorithm the use of martingale basis functions allows to compute 
some conditional expectations in the approximate backward dynamic program explicitly. These closed form computations can be thought of as a perfect control variate 
within the regression algorithm. 

For the computation of the upper confidence bounds we use the explicit recursion for $\theta^{up}$ derived in Example \ref{exmp:thetafunding}. For the computation of the lower confidence 
bound we note that the defining equation (\ref{defcontr}) for the approximate controls $(\tilde r, \tilde \rho)$ for the lower bound can be solved explicitly as
\begin{eqnarray*}
\tilde r_i&=&-R^b{\bf 1}_{\{\tilde y_i(X_i)\leq \sigma^{-1} \sum_{d=1}^D \tilde z_{d,i}(X_i)\}} -R^l{\bf 1}_{\{\tilde y_i(X_i)> \sigma^{-1} \sum_{d=1}^D \tilde z_{d,i}(X_i)\}}  \\
\tilde \rho_{d,i}&=&  -\sigma^{-1}(\tilde r_i+\mu).
\end{eqnarray*}
Figure 1 illustrates the effectiveness of the control variate for the inner samples in the computation of the upper bounds for the European option case with $n=40$ time steps. The input approximation is 
generated by the martingale basis algorithm with seven basis functions and $\Lambda^{reg}=1,000$ sample paths for the empirical regression. The figure depicts the corresponding upper bound estimator 
for the option price $Y_0$ with $\Lambda^{out}=10,000$ sample paths as a function of the number of inner samples $\Lambda^{in}$. From top to bottom, it shows the upper estimators 
$\hat Y^{up,AB}$ (i.e. without inner control variate), $\hat Y^{up,ABC}$ (i.e. with inner control variate), and for comparison the lower bound estimator $\hat Y^{low,AB}$.

\begin{figure}
\centerline{\rotatebox[origin=c]{0}{
\includegraphics[scale=0.5]{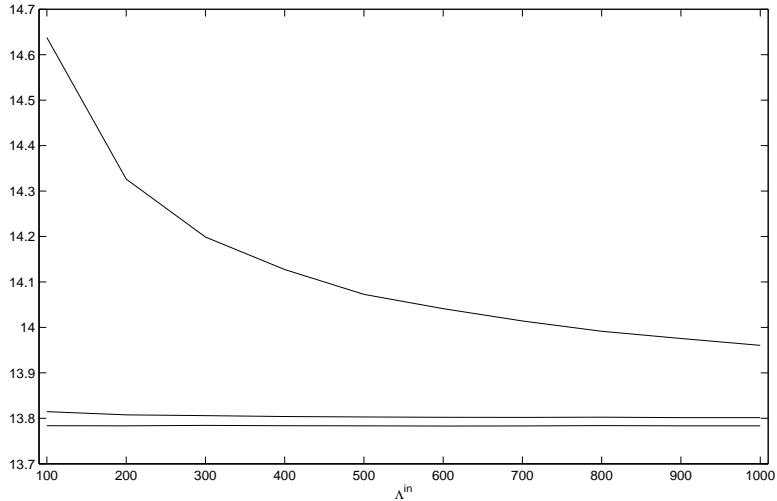}}}
 \caption{Influence of the number of inner simulations and the control variate: upper bound without inner control variate, upper bound with inner control variate, and lower bound (from the top to the bottom).}
\end{figure}

We immediately observe that the predominant part of the upper bias in $\hat Y^{up,AB}$ stems from the subsampling in the approximate construction of the Doob martingales. Without the use of inner control variates,
the relative error between upper and lower estimator is about 6\% for $\Lambda^{in}=100$ inner samples and decreases to about 1.5\% for   $\Lambda^{in}=1,000$ inner samples. Application 
of the inner control variates reduces this relative error to less than 0.25\% even in the case of only $\Lambda^{in}=100$ inner samples.

Table 1 illustrates the influence of different input approximations. It shows realizations of the lower estimator $\hat Y^{low,ABC}$ and the upper estimator $\hat Y^{up,ABC}$ for the option price $Y_0$ as well as the empirical
standard deviations, as the number of time steps increases from $n=40$ to $n=160$.  The column on the left explains which algorithm is run for the input approximation. Here,
LGW stands for the Lemor-Gobet-Warin algorithm and MB for the martingale basis algorithm. It also states the number of regression samples and the number of basis functions $b^y$, which are applied 
in the least-squares Monte Carlo algorithms. The lower and upper price estimates for the Bermudan option case are presented in the last two lines. In this case, the martingale basis algorithm is not available, and the Lemor-Gobet-Warin algorithm 
is run with the six basis functions stated above. We apply $\Lambda^{out}=10,000$ and $\Lambda^{in}=100$ samples in all cases.

\begin{table}
\begin{center}
\begin{tabular}{|c|c|c| c| c|}
\hline
Algorithm  \textbackslash $\;\;n$	& 40 & 80 & 120 & 160\\
\hline
$\twertt{\displaystyle LGW}{\Lambda^{reg}=10^4, b^y=2}$ &  $ \twert{ 13.7786 }{   0.0028  } \; \twert{ 13.8339 }{   0.0031 }$ &  $\twert{13.7597}{    0.0033} \; \twert{13.8858 }{   0.0041 }$ &$\twert{ 13.7583  }{  0.0037 }\; \twert{  13.9482}{  0.0051}$ & $\twert{ 13.7478}{    0.0043}\;  \twert{ 14.0149 }{   0.0062}$\\
\hline
$\twertt{\displaystyle LGW}{\Lambda^{reg}=10^5, b^y=2}$ & $ \twert{13.7783 }{0.0022 } \; \twert{13.8172 }{0.0024 }$ & $ \twert{ 13.7817 }{0.0022 } \; \twert{13.8443 }{0.0027 }$ &  $ \twert{13.7848 }{ 0.0024 } \; \twert{13.8682  }{0.0029 }$ & $ \twert{13.7855 }{0.0025 } \; \twert{13.8967 }{0.0033 }$ \\ 
\hline
$\twertt{\displaystyle MB}{\Lambda^{reg}=10^2, b^y=2}$ & $ \twert{13.7850 }{0.0022 } \; \twert{13.8185 }{0.0023 }$ & $ \twert{13.7898 }{0.0021 } \; \twert{13.8435 }{ 0.0025 }$ &  $ \twert{ 13.7863 }{0.0022 } \; \twert{13.8578 }{0.0025 }$  & $ \twert{13.7904 }{0.0022 } \; \twert{13.8779 }{0.0026 }$ \\
\hline
$\twertt{\displaystyle LGW}{\Lambda^{reg}=10^5, b^y=7}$ &    $ \twert{13.7818 }{0.0020  } \; \twert{13.8140 }{0.0021 }$ & $ \twert{13.7767 }{0.0020 } \; \twert{13.8321 }{0.0022 }$ &  $ \twert{ 13.7789 }{ 0.0022 } \; \twert{13.8560 }{ 0.0025 }$   & $ \twert{13.7764 }{0.0025 } \; \twert{13.8902 }{0.0031 }$  \\
\hline
$\twertt{\displaystyle LGW}{\Lambda^{reg}=10^6, b^y=7}$ &    $ \twert{13.7829 }{0.0017  } \; \twert{13.8079 }{0.0018 }$ & $ \twert{13.7867 }{0.0016 } \; \twert{13.8233 }{0.0018 }$ &  $ \twert{ 13.7884 }{ 0.0017 } \; \twert{13.8393 }{ 0.0020 }$   & $ \twert{13.7867 }{0.0017 } \; \twert{13.8515 }{0.0022 }$  \\
\hline
$\twertt{\displaystyle MB}{\Lambda^{reg}=10^3, b^y=7}$ & $ \twert{ 13.7844 }{0.0017 } \; \twert{13.8077 }{ 0.0017 }$ & $ \twert{13.7897 }{0.0016 } \; \twert{13.8245 }{ 0.0017 }$ &  $ \twert{ 13.7887  }{ 0.0016 } \; \twert{ 13.8353 }{ 0.0019 }$ & $ \twert{13.7880 }{ 0.0017 } \; \twert{13.8485 }{0.0021 }$ \\
\hline
\hline
$\twertt{\displaystyle LGW{\scriptstyle \textnormal{Bermudan}}}{\Lambda^{reg}=10^5}$ &  $\twert{15.5362 }{0.0028 } \; \twert{15.5664 }{0.0028 }$ & $\twert{15.5441 }{0.0037 } \; \twert{15.6160 }{0.0035 }$ & $\twert{15.5246  }{0.0041 } \; \twert{15.6396 }{0.0042 }$ & $\twert{15.5342 }{0.0041 } \; \twert{15.6886 }{0.0048   }$\\
\hline
$\twertt{\displaystyle LGW{\scriptstyle \textnormal{Bermudan}}}{\Lambda^{reg}=10^6}$ &  $\twert{15.5422 }{0.0028 } \; \twert{15.5684 }{0.0026 }$ & $\twert{15.5482 }{0.0032 } \; \twert{15.6050 }{0.0033 }$ & $\twert{15.5441  }{0.0035 } \; \twert{15.6364 }{0.0039 }$ & $\twert{15.5443 }{0.0039 } \; \twert{15.6694 }{0.0042   }$\\
\hline
\end{tabular}
\end{center}
\caption{Upper and lower price bounds for different time discretizations and input approximations in the European and Bermudan case. Standard deviations are in brackets.}
\end{table}

By and large, the table shows that in this 5-dimensional example extremely tight 95\% confidence intervals can be computed by the primal-dual algorithm, although the input approximations are based on very few, but well chosen,
basis functions. For the martingale basis algorithm as input approximation with just two basis functions and 100 regression paths the relative error between lower and upper 95\%-confidence bound 
is about 0.7\% even for $n=160$ steps in the time discretization. It can be further decreased to less than 0.5\%, when seven basis functions and 1,000 regression paths are applied. 
If one takes the input approximation of the Lemor-Gobet-Warin algorithm with the same set of basis functions, then the primal-dual algorithm can in principle produce confidence intervals 
of about the same length as in the case of the martingale basis algorithm. However, in our simulation study the number of regression paths must be increased by a factor of 1,000 in order 
to obtain input approximations which have the same quality as those computed by the martingale basis algorithm. Hence our numerical results demonstrate the huge variance reduction effect of the martingale basis 
algorithm. In the Bermudan option case, the primal-dual algorithm still yields 95\%-confidence intervals with a relative width of less than 1\% for up to $n=160$ time steps, when the input 
approximation is computed by the Lemor-Gobet-Warin algorithm with 6 basis functions and 1 million regression paths.

\section{The case of a non-convex generator}

In this section we drop the assumption on the convexity of the generator $f$ and merely assume that the standing assumptions are in force. In this situation the construction
of confidence bounds for $Y_0$ can be based on local approximations of $f$ by convex and concave generators.

\subsection{Upper bounds}

We first turn to the construction of upper bounds.
For fixed $i=0,\ldots,n-1$ we assume that some approximation $(\tilde Y_j, \tilde Z_j)_{j=i,\ldots,n-1}$ of $(Y_j, E_j[\beta_{j+1}Y_{j+1}])_{j=i,\ldots,n-1}$ is given. 
This approximation can be pre-computed by any algorithm. 
We merely assume that the approximation is adapted and satisfies
$$
E\left[\sum_{j=i}^{n-1} \left(|\tilde Y_j|+ \sum_{d=1}^D \alpha^{(d)}_j|\tilde Z_{d,j}|\right) \Delta_j\right]<\infty.
$$
The set of such admissible input approximations is denoted by $\mathcal{A}_i$.

We now choose
a measurable function
$$
h^{up}: \Omega \times \{0,\ldots,n\} \times \mathbb{R} \times \mathbb{R}^D\times \mathbb{R} \times \mathbb{R}^D \rightarrow \mathbb{R}
$$
with the following properties:
\begin{itemize}
 \item[a)] $h^{up}(\cdot,\tilde y,\tilde z;y,z)$ is adapted for every $(\tilde y,\tilde z),(y,z)\in \mathbb{R} \times \mathbb{R}^D$.
Moreover $h^{up}$ satisfies the stochastic Lipschitz condition
$$
|h^{up}(i,\tilde y,\tilde z;y,z)-h^{up}(i,\tilde y,\tilde z;y',z')|\leq \alpha^{(0)}_i |y-y'|+\sum_{d=1}^D  \alpha^{(d)}_i |z_d-z_d'|
$$
for every $(\tilde y,\tilde z), (y,z), (y',z') \in \mathbb{R} \times \mathbb{R}^D$ (with the same stochastic Lipschitz constants as $f$).
\item[b)] $h^{up}(i,\tilde y,\tilde z;y,z)$ is convex in $(y,z)$, $h^{up}(i,\tilde y,\tilde z;0,0)=0$ for every $(\tilde y,\tilde z)\in \mathbb{R} \times \mathbb{R}^D$, and
$$
h^{up}(i,\tilde y,\tilde z;\tilde y-y,\tilde z-z)\geq f(i,y,z)-f(i,\tilde y,\tilde z)
$$
for every $(\tilde y,\tilde z),(y,z)\in \mathbb{R} \times \mathbb{R}^D$.
\end{itemize}
\begin{rem}
 Given $h^{up}$ and the approximation  $(\tilde Y_i, \tilde Z_i)$ we can define a new generator
$$
f^{up}(i,y,z):=f(i,\tilde Y_i, \tilde Z_i)+h^{up}(i,\tilde Y_i,\tilde Z_i; \tilde Y_i-y, \tilde Z_i-z).
$$
Then $f^{up}(i,y,z)$ is convex in $(y,z)$ and dominates the original generator $f$, i.e.   $f^{up}(i,y,z)\geq f(i,y,z)$. Moreover,
\begin{eqnarray*}
 && E[|f^{up}(i,Y_i,E_i[\beta_{i+1}Y_{i+1}])-f(i,Y_i,E_i[\beta_{i+1}Y_{i+1}])|]\\ &\leq&  2E\left[\alpha^{(0)}_i |\tilde Y_i-Y_i|+\sum_{d=1}^D  \alpha^{(d)}_i |\tilde Z_i-E_i[\beta_{i+1}Y_{i+1}]|\right],
\end{eqnarray*}
which shows that -- evaluated at the true solution $(Y_i, E_i[\beta_{i+1}Y_{i+1}])$ -- the auxiliary generator $f^{up}$ approximates the true generator $f$, as the approximation $(\tilde Y,\tilde Z)$ approaches the true 
solution.
\end{rem}

A generic choice is the function
$$
h^{|up|}(i,\tilde y, \tilde z;y,z)= \alpha^{(0)}_i |y|+\sum_{d=1}^D  \alpha^{(d)}_i |z_d|,
$$
which obviously satisfies the properties a) and b) above. We will illustrate in the numerical examples below, that it might be beneficial to tailor the function $h^{up}$ to the specific problem instead of applying the generic choice
$h^{|up|}$.

Given $h^{up}$, $(\tilde Y, \tilde Z)$ we define $\Theta_i^{h^{up}}=\Theta^{h^{up}}_i(\tilde Y,\tilde Z) $ via
\begin{eqnarray}\label{pathwise2}
\Theta^{h^{up}}_i&=&\max\{S_i,\Theta^{h^{up}}_{i+1}-(\tilde Y_{i+1}-E_i[\tilde Y_{i+1}]) +f_i(\tilde Y_i, \tilde Z_i)\Delta_i \nonumber \\ && +h^{up}(i,\tilde Y_i,\tilde Z_i;\tilde Y_i-\Theta^{h^{up}}_i,\tilde Z_i- \beta_{i+1} \Theta_{i+1}^{h^{up}}+ \beta_{i+1} \tilde Y_{i+1}- E_i[\beta_{i+1} \tilde Y_{i+1}]) \Delta_i\},
\end{eqnarray}
initiated at $\Theta^{h^{up}}_n=S_n$. We then obtain the following minimization problem with value process $Y_i$ in terms of $\Theta^{h^{up}}_i(\tilde Y,\tilde Z)$.
\begin{thm}\label{thm:dualgeneric}
 For every $i=0,\ldots,n$,
\begin{eqnarray*}
 Y_i=\essinf_{(\tilde Y,\tilde Z) \in \mathcal{A}_i} E_i[\Theta^{h^{up}}_i(\tilde Y,\tilde Z)].
\end{eqnarray*}
Moreover, a minimizing pair is given by $(Y^*_j, Z^*_j)=(Y_j,E_j[\beta_{j+1}Y_{j+1}])$ which even satisfies the principle of pathwise optimality.
\end{thm}
 \begin{proof}
We fix a pair of adapted and integrable processes $(\tilde Y,\tilde Z)$ and
define $Y^{up}_j$, $j\geq i$, as
$$
Y^{up}_j=\max\{S_j,E_j[Y^{up}_{j+1}] +[f(j,\tilde Y_j, \tilde Z_j)+h^{up}(j,\tilde Y_j,\tilde Z_j; \tilde Y_j-Y^{up}_j, \tilde Z_j-E_j[\beta_{j+1}Y^{up}_{j+1}])]\Delta_j\},\quad Y^{up}_n=S_n,
$$
which satisfies $Y^{up}_i\geq Y_i$ by the comparison result in Proposition \ref{prop:comparison}. Then, an application of Theorem \ref{thm:convexdual}, with $Y_i$ replaced by
$Y^{up}_i$ yields $E_i[\Theta^{h^{up}}_i(\tilde Y,\tilde Z)]\geq Y^{up}_i$. Hence,
$$
Y_i\leq \essinf_{(\tilde Y,\tilde Z) \in\mathcal{A}_i} E_i[\Theta^{h^{up}}_i(\tilde Y,\tilde Z)].
$$
It now suffices to show that
$$
Y_j = \Theta^{h^{up}}_j( Y_\cdot, E_\cdot[\beta_{\cdot+1} Y_{\cdot+1}])=:\Theta^{h^{up},*}_j,
$$
$P$-almost surely for every $j=i,\ldots,n$. This is certainly true for $j=n$. Going backwards in time we obtain by induction
\begin{eqnarray*}
 &&\Theta^{h^{up},*}_j \\&=&\max\{S_j,Y_{j+1}-(Y_{j+1}-E_j[Y_{j+1}]) +f(j,Y_j, E_j[\beta_{j+1}Y_{j+1}])\Delta_j \nonumber \\ && +
h^{up}(j,Y_j,E_j[\beta_{j+1}  Y_{j+1}];Y_j-\Theta^{h^{up},*}_j,E_j[\beta_{j+1}Y_{j+1}]- \beta_{j+1} Y_{j+1}+ \beta_{j+1} Y_{j+1}- E_j[\beta_{j+1}  Y_{j+1}]) \Delta_j\}
\\ &=&  \max\{S_j,E_j[Y_{j+1}] +(f(j,Y_j, E_j[\beta_{j+1}Y_{j+1}])+
h^{up}(j,Y_j,E_j[\beta_{j+1}  Y_{j+1}];Y_j-\Theta^{h^{up},*}_j,0)) \Delta_j\}
\end{eqnarray*}
As $h^{up}(j,Y_j,E_j[\beta_{j+1}  Y_{j+1}];0,0)=0$, we observe that $Y_j$ also solves the above equation. Hence, by uniqueness (due to the Lipschitz assumption on $h^{up}$), we obtain
$Y_j=\Theta^{h^{up},*}_j$.
 \end{proof}

 \subsection{Lower bounds}

A maximization problem with value process $Y_i$ can be constructed analogously by bounding $f$ from below by a concave generator. The main difference is that in place of the results of Section 3
we now rely on the following  result for the concave case which is proved at the end of this section:
\begin{thm}\label{thm:concave}
  Suppose $f$ is concave in $(y,z)$. \\
(i) Then, for every $i=0,\ldots,n$, 
\begin{eqnarray*}
&& Y_i= \essinf_{M^0 \in \mathcal{M}_1}  \essinf_{(r,\rho) \in \mathcal{U}_i((-f)^\#)} E_i[\vartheta^{up}_i(r,\rho,M^0)  ], \qquad\text{where }\\
&& \vartheta^{up}_i(r,\rho,M^0) =  \max_{k=i,\ldots,n} \Gamma_{i,k}(-r,-\rho)S_k  + \sum_{j=i}^{k-1} \Gamma_{i,j}(-r,-\rho) \frac{(-f)^\#(j,r_j,\rho_j)\Delta_j}{1+r_j \Delta_j}-(M^0_k-M^0_i).
\end{eqnarray*}
Minimizers (even in the sense of pathwise optimality) are given by 
$(r_j^*,\rho_j^*)_{j\geq i}$  satisfying 
\begin{equation}
-r^*_jY_j-{\rho^*_j}^\top E_j[\beta_{j+1}Y_{j+1}]+(-f)^\#(j,r^*_j,\rho^*_j)=f(j,Y_j, E_j[\beta_{j+1}Y_{j+1}]) \label{optimal_concave}
\end{equation}
and $M^{0,*}$ being the martingale part of the Doob decomposition of $(Y_j \Gamma_{i,j}(-r^*,-\rho^*))_{j\geq i}$.
\\ (ii) Given a stopping time $\tau\in \bar{\mathcal{S}}_i$ and a martingale $(M^0,M)\in \mathcal{M}_{1+D}$, define ${\vartheta}^{low}_j={\vartheta}^{low}_j(\tau, M^0,M)$ 
for $i\leq j <\tau$
via 
$$
{\vartheta}^{low}_j={\vartheta}^{low}_{j+1}-(M^0_{j+1}-M^0_j) +f(j,{\vartheta}^{low}_j, \beta_{j+1} {\vartheta}^{low}_{j+1}-(M_{j+1}-M_j) )\Delta_j,\quad {\vartheta}^{low}_\tau=S_\tau.
$$
Then,
\begin{eqnarray*}
 Y_i&=&\esssup_{\tau \in \bar{\mathcal{S}}_i}  \esssup_{ (M^0, M)\in \mathcal{M}_{1+D}} E_i[{\vartheta}^{low}_i(\tau, M^0,M)]
\end{eqnarray*}
A maximizer (even in the sense of pathwise optimality) is given by the triplet $(\tau^*_i, M^{0,*},  M^{*})$, where $\tau^*_i$ was defined in (\ref{tau*}) and $ M^{0,*},\;  M^{*}$ are the Doob martingales 
of $Y$ and $\beta Y$, respectively. \\
\end{thm}
This result is not completely symmetric to the convex case, because the reflection at a lower barrier
 (i.e. application of the maximum-operator) is  convex. Note that if $f$ is concave itself then the upper and lower bounds from  Theorem \ref{thm:concave}
are preferable to the upper bound of Theorem \ref{thm:dualgeneric} and to the generic lower bounds which are constructed next.

We denote by $h^{low}$ any mapping which satisfies the same properties as $h^{up}$ but with
condition b) replaced by
\begin{itemize}
 \item[b')] $h^{low}(i,\tilde y,\tilde z;y,z)$ is concave in $(y,z)$, $h^{low}(i,\tilde y,\tilde z;0,0)=0$ for every $(\tilde y,\tilde z)\in \mathbb{R} \times \mathbb{R}^D$, and
$$
h^{low}(i,\tilde y,\tilde z;\tilde y-y,\tilde z-z)\leq f(i,y,z)-f(i,\tilde y,\tilde z)
$$
for every $(\tilde y,\tilde z),(y,z)\in \mathbb{R} \times \mathbb{R}^D$.
\end{itemize}
The generic choice is now
$$
h^{|low|}(i,\tilde y, \tilde z;y,z)= -\alpha^{(0)}_i |y|-\sum_{d=1}^D  \alpha^{(d)}_i |z_d|.
$$
Given $h^{low}$, a pair of adapted processes $(\tilde Y, \tilde Z)$ and a stopping time $\tau \in \bar{\mathcal{S}}_0$ we define $\Theta_i^{h^{low}}=\Theta^{h^{low}}_i(\tilde Y,\tilde Z,\tau) $ via
\begin{eqnarray}\label{pathwise3}
\Theta^{h^{low}}_i&=&\Theta^{h^{low}}_{i+1}-(\tilde Y_{i+1}-E_i[\tilde Y_{i+1}]) +f_i(\tilde Y_i, \tilde Z_i)\Delta_i \nonumber \\ && +h^{low}(i,\tilde Y_i,\tilde Z_i;\tilde Y_i-\Theta^{h^{low}}_i,\tilde Z_i- \beta_{i+1} \Theta_{i+1}^{h^{low}}+ \beta_{i+1} \tilde Y_{i+1}- E_i[\beta_{i+1} \tilde Y_{i+1}]) \Delta_i,
\end{eqnarray}
for $i<\tau$ initiated at $\Theta^{h^{low}}_\tau=S_\tau$.
Making use of Theorem \ref{thm:concave} and the same arguments as in Theorem \ref{thm:dualgeneric} we obtain:
\begin{thm}\label{thm:dualgeneric2}
  For every $i=0,\ldots,n$,
\begin{eqnarray*}
 Y_i=\esssup_{\tau \in \bar{\mathcal{S}}_i}\esssup_{(\tilde Y,\tilde Z) \in \mathcal{A}_i} E_i[\Theta^{h^{low}}_i(\tilde Y,\tilde Z,\tau)].
\end{eqnarray*}
Moreover, a minimizing triplet is given by $(Y^*_j, Z^*_j,\tau^*)=(Y_j,E_j[\beta_{j+1}Y_{j+1}],\tau^*_i)$ which even satisfies the principle of pathwise optimality. (We recall that
$\tau^*_i$ was defined in (\ref{tau*})).
\end{thm}

\begin{exmp}\label{exmp:generic}
 For the generic choices $h^{|up|}$ and $h^{|low|}$, we can apply Proposition \ref{prop:max} in order to make the recursion formulas in (\ref{pathwise2}) and (\ref{pathwise3}) explicit.
They read
\begin{eqnarray*}
 \Theta^{h^{|up|}}_i&=&\max\Bigl\{S_i,\sup_{r\in\{-\alpha^{(0)}_i,\alpha^{(0)}_i\}} \frac{1}{1+r\Delta_i}\Bigl(\Theta^{h^{|up|}}_{i+1}-(\tilde Y_{i+1}-E_i[\tilde Y_{i+1}])
+f(i,\tilde Y_i, \tilde Z_i)\Delta_i \nonumber \\ && +\sum_{d=1}^D |\tilde Z_{d,i}- \beta_{d,i+1} \Theta_{i+1}^{h^{|up|}}+ \beta_{d,i+1} \tilde Y_{i+1}- E_i[\beta_{d,i+1} \tilde Y_{i+1}]|\Delta_i \Bigr) \Bigr\},
\end{eqnarray*}
and
\begin{eqnarray*}
 \Theta^{h^{|low|}}_i&=&\inf_{r\in\{-\alpha^{(0)}_i,\alpha^{(0)}_i\}} \frac{1}{1+r\Delta_i}\Bigl(\Theta^{h^{|low|}}_{i+1}-(\tilde Y_{i+1}-E_i[\tilde Y_{i+1}])
+f(i,\tilde Y_i, \tilde Z_i)\Delta_i \nonumber \\ && -\sum_{d=1}^D \alpha^{(d)}_i\ |\tilde Z_{d,i}- \beta_{d,i+1} \Theta_{i+1}^{h^{|low|}}+ \beta_{d,i+1} \tilde Y_{i+1}- E_i[\beta_{d,i+1} \tilde Y_{i+1}]|\Delta_i \Bigr) .
\end{eqnarray*}
The main advantage of the corresponding upper and lower bounds is that they can be calculated generically without any extra information on $f$ (such as the convex conjugates which were
required in the section on convex generators). There is, however, a price to pay for this generic approach. Indeed, given the Lipschitz process $\alpha^{(d)}_i$, the choice
$h^{|up|}$, $h^{|low|}$ can be shown to lead to the crudest upper and lower bounds among all admissible functions $h^{up}$, $h^{low}$, i.e.
$$
E_i[\Theta^{h^{|up|}}_i(\tilde Y,\tilde Z)]\geq E_i[\Theta^{h^{up}}_i(\tilde Y,\tilde Z)]
$$
for every pair $(\tilde Y, \tilde Z)\in \mathcal{A}_i$, and analogously for the lower bounds.
 In practice, the generic bounds may be too crude, when $D$ is large
and the approximation $\tilde Z_j$ of $E_j[\beta_{j+1}Y_{j+1}]$ is not yet very good. In general we therefore recommend to choose the functions $h^{up}$ and $h^{low}$ in a way that
$h^{up}(j,\tilde Y_j,\tilde Z_j;y,z)$ and $h^{low}(j,\tilde Y_j,\tilde Z_j;y,z)$ 
are close
to zero in a neighborhood of zero in the $(y,z)$-coordinates, in which one expects the residuals $(\tilde Y_j-Y_j, \tilde Z_j-E_j[\beta_{j+1}Y_{j+1}])$ to be
typically located.
\end{exmp}

We close this section with the proof of Theorem \ref{thm:concave}.

\begin{proof}[Proof of Theorem \ref{thm:concave}]
 (i) Given $(r,\rho)\in \mathcal{U}_i((-f)^\#)$ and $j=i,\ldots,n$, $k=j,\ldots,n$, define 
$$
Y_j(k,r,\rho)=E_j\left[\Gamma_{j,k}(-r,-\rho)S_k  + \sum_{l=j}^{k-1} \Gamma_{j,l}(-r,-\rho) \frac{(-f)^\#(l,r_l,\rho_l)\Delta_l}{1+r_l \Delta_l}\right]
$$
Then, the optional sampling theorem yields for every stopping time $\tau\in \bar{\mathcal{S}}_i$ and every martingale $M^0\in \mathcal{M}_1$
\begin{eqnarray*}
 Y_i(\tau,r,\rho)&=&E_i\left[\Gamma_{i,\tau}(-r,-\rho)S_\tau  + \sum_{l=i}^{\tau-1} \Gamma_{j,l}(-r,-\rho) \frac{(-f)^\#(l,r_l,\rho_l)\Delta_l}{1+r_l \Delta_l}-(M^0_\tau-M^0_i)\right]\\
&\leq&  E_i[\vartheta^{up}_i(r,\rho,M^0)].
\end{eqnarray*}
The same argument as in the first part of the proof of Theorem \ref{thm:primal} now shows by concavity that
$
Y_i(\tau,r,\rho)\geq Y_i.
$
Hence,
$$
E_i[\vartheta^{up}_i(r,\rho,M^0)]\geq Y_i.
$$
Now we denote the Doob martingale of $Y_j \Gamma_{i,j}(-r^*,-\rho^*)$ by $M^{0,*}_j$ for $j=i,\ldots, n$, and choose a pair $(r^*,\rho^*)\in \mathcal{U}_i((-f)^\#)$ which satisfies (\ref{optimal_concave}). Such a pair exists again by Lemma \ref{lem:subgradient}. 
Define
$$
\vartheta_j^*:=\vartheta^{up}_j(r^*,\rho^*,M^{0,*}),\quad j=i,\ldots, n.
$$
We show by induction on $j=n,\ldots,i$, that $Y_j=\vartheta_j^*$. Note first that $Y_n=S_n=\vartheta_n^*$. In order to prove the claim for $j=i,\ldots,n-1$
we first observe that 
\begin{eqnarray*}
 M^{0,*}_k-M^{0,*}_j&=&\sum_{l=j}^{k-1} \left(Y_{l+1} \Gamma_{i,l+1}(-r^*,-\rho^*)-E_l[Y_{l+1} \Gamma_{i,l+1}(-r^*,-\rho^*)]\right) \\
&=& \sum_{l=j}^{k-1} \Gamma_{i,l}\frac{Y_{l+1}-E_l[Y_{l+1}]-(\rho^*_l)^\top(\beta_{l+1}Y_{l+1}-E_{l}[\beta_{l+1}Y_{l+1}])\Delta_l}{1+r^*_l\Delta_l}
\end{eqnarray*}
for $k=j,\dots, n$.
Hence, 
\begin{eqnarray*}
 \vartheta_j^*&=&  \max_{k=j,\ldots,n} \Biggl( \Gamma_{j,k}(-r^*,-\rho^*)S_k  \\&&+ \sum_{l=j}^{k-1} \Gamma_{j,l}(-r^*,-\rho^*) \frac{((-f)^\#(l,r^*_l,\rho^*_l)+(\rho^*_l)^\top(\beta_{l+1}Y_{l+1}-E_{l}[\beta_{l+1}Y_{l+1}])\Delta_l
-Y_{l+1}+E_l[Y_{l+1}]}{1+r^*_l \Delta_l}\Biggr) \\ &=&
\max\{S_j, \vartheta_{j+1}^*-Y_{j+1}+E_j[Y_{j+1}]\\ &&\quad+(-r^*_j\vartheta_j^*-(\rho^*_j)^\top \beta_{j+1}\vartheta^*_{j+1} + (-f)^\#(j,r^*_j,\rho^*_j)+(\rho^*_j)^\top(\beta_{j+1}Y_{j+1}-E_{j}[\beta_{j+1}Y_{j+1}]))\Delta_j\}.
\end{eqnarray*}
By the induction hypothesis and (\ref{optimal_concave}) we obtain 
$$
\vartheta_j^*=\max\{S_j,E_j[Y_{j+1}]+(f(j,Y_j, E_j[\beta_{j+1}Y_{j+1}])+r^*_j(Y_j-\vartheta_j^*))\Delta_j\}.
$$
As $Y_j$ is the unique solution of this equation we conclude that $\vartheta^*_j=Y_j$.
\\[0.1cm]
(ii) Fix $(M^0,M)\in \mathcal{M}_{1+D}$ and $\tau \in \bar{\mathcal{S}_i}$. Then, by concavity of $f$ we observe analogously to the proof of Theorem \ref{thm:convexdual} that
$E_j[{\vartheta}^{low}_j(\tau, M^0,M)]]$, $i\leq j\leq \tau$, is a subsolution to the nonreflected BSDE with generator $f$ and terminal time $\tau$. 
The solution of the latter BSDE was denoted by $Y^{(\tau)}_j$ in Proposition \ref{prop:stopping}.  
Hence, by Propositions \ref{prop:comparison} and \ref{prop:stopping}
$$
E_i[{\vartheta}^{low}_i(\tau, M^0,M)]]\leq Y^{(\tau)}_i \leq Y_i.
$$
In order to prove pathwise optimality of $(\tau^*_i, M^{0,*},  M^{*})$ one proceeds as in the proof of Theorem \ref{thm:convexdual}. The analogous induction argument shows that 
for $j=i,\ldots, \tau_i^*-1$
$$
{\vartheta}^{low}_j(\tau^*_i, M^{0,*},M^*)=E_j[Y_{j+1}]+f(j, {\vartheta}^{low}_j(\tau^*_i, M^{0,*},M^*), E_j[\beta_{j+1}Y_{j+1}])\Delta_j,
$$
which again, by the Lipschitz continuity of $f$, implies ${\vartheta}^{low}_j(\tau^*_i, M^{0,*},M^*)=Y_j$.
\end{proof}

\subsection{Numerical examples}

Once the functions $h^{low}$ and $h^{up}$ are chosen, an algorithm for computing confidence intervals for $Y_0$ based on Theorems \ref{thm:dualgeneric} and \ref{thm:dualgeneric2} can be designed 
analogously to the primal-dual algorithm in Section \ref{sec:algorithm} for the convex case. 

We first illustrate the algorithm in the context of Example \ref{exmp:intro} (ii). For the underlying, we choose the same five-dimensional geometric Brownian motion as  in Section \ref{sec:num1} except that $T=1$ and the drift and risk-free rate equal $R=0.02$. The payoff of the (European) claim is given by $G_n(x)=\min_{d=1,\ldots,D}x_d$. 
For the default risk function $Q$, we assume that there are three regimes, high risk, intermediate risk and low risk: There are thresholds $v^h<v^l$ and rates $\gamma^h> \gamma^l$ such that $Q(y)=\gamma^h$ for $y<v^h$ and $Q(y)=\gamma^l$ for $y>v^l$. Over $[v^h,v^l]$, $Q$ interpolates linearly.
The resulting function $f$ is Lipschitz continuous but generally neither convex nor concave. The candidates for the Lipschitz constant $\alpha^{(0)}$ are the absolute values of the left and right derivatives of $f$ in $v^h$ and $v^l$.
In the implementation, we stick to the generic choice 
$$
- h^{|low|}(i,\tilde y;y)= h^{|up|}(i,\tilde y;y)=\alpha^{(0)}|y|,
$$
using that the nonlinearity is independent of the $Z$-part in this example. We choose 
$$ v^h = 54,\quad v^l=90,\quad  \gamma^h=0.2,\quad \gamma^l=0.02 . $$ 
For the calculation of $\tilde{y}$, we use the Lemor-Gobet-Warin algorithm with two basis functions, $1$ and $E[G_n(X_n)|X_i=x]$, and $\Lambda^{reg}=100,000$.  Moreover, $\Lambda^{out}=4,000, \Lambda^{in}=1,000$.

In the absence of default risk, the claim's value is given by $78.37$. Table 2 displays upper and lower price bounds for different time discretizations and recovery rates $\delta$. As expected, a smaller recovery rate leads to a smaller option value.
The relative width of the  confidence intervals is well below $0.5 \%$ in all cases. For the larger values of $\delta$, the bounds are even tighter: 
Larger values of $\delta$ lead to less nonlinearity in the pricing problem and to smaller Lipschitz constants ($\alpha^{(0)}=0.41,0.27,0.12$ for $\delta=0,1/3,2/3$). 
Compared to the example of Section \ref{sec:num1}, the bounds are much less dependent on the time discretization. This is due to the fact, that no $Z$-part has to be approximated, as is the case for many BSDEs in the credit risk literature, see  \citet{crepey2013counterparty, henry2012counterparty}. To sum up, the generic approach is perfectly sufficient in this example.

\begin{table}
\begin{center}
\begin{tabular}{|c|c|c| c| c|}
\hline
 $\; \delta$  \textbackslash $\;\;n$	& 40 & 80 & 120 & 160\\
\hline
$0$ & $ \twert{  71.6551 }{0.0071  } \; \twert{ 71.8589 }{0.0068  }$ & $ \twert{ 71.6774  }{ 0.0072  } \; \twert{ 71.8828   }{ 0.0068   }$ & $ \twert{71.6664  }{ 0.0070   } \; \twert{71.8656   }{ 0.0068   }$ & $ \twert{ 71.6621  }{0.0069  } \; \twert{  71.8659 }{  0.0072 }$  \\
\hline
$\frac13$ &   $ \twert{74.1023  }{  0.0062  } \; \twert{ 74.2241   }{  0.0060  }$ & $ \twert{ 74.1010  }{  0.0065 } \; \twert{ 74.2225   }{  0.0062   }$ & $ \twert{ 74.1032  }{  0.0062 } \; \twert{  74.2229  }{  0.0061 }$ & $ \twert{74.1187   }{  0.0065 } \; \twert{74.2391  }{ 0.0063  }$   \\ 
\hline
$\frac23$ &   $ \twert{ 76.3335   }{ 0.0057  } \; \twert{ 76.3865   }{ 0.0057  }$ & $ \twert{ 76.3364 }{0.0057  } \; \twert{ 76.3886   }{ 0.0057  }$ & $ \twert{ 76.3416 }{ 0.0059  } \; \twert{76.3943  }{ 0.0058 }$ & $ \twert{ 76.3290 }{ 0.0061   } \; \twert{ 76.3814   }{0.0059  }$  \\ 
\hline
\end{tabular}
\end{center}
\caption{Upper and lower price bounds for different recovery rates and time discretizations. Standard deviations are in brackets.}
\end{table}

We finally revisit the example of Section \ref{sec:num1}. For the input approximation we run the martingale basis algorithm with seven basis functions for $Y$ and 1,000 regression paths as specified there. 
The confidence bounds for the European call spread option on the maximum of five Black-Scholes stocks are calculated with $\Lambda^{in}=\Lambda^{out}=1,000$ paths based on the following choices 
of $h^{low}$ and $h^{up}$. For the \emph{fully generic} implementation we apply
$$
-h^{|low|}(i,\tilde y,\tilde z;y,z)=h^{|up|}(i,\tilde y,\tilde z;y,z)=R^b |y|+\frac{\max\{|R^b-\mu|,|R^l-\mu|\} }{\sigma}\sum_{d=1}^5 |z_d| .
$$
 For the \emph{semi-generic} implementation we choose
\begin{eqnarray*}
 h^{low}(i,\tilde y,\tilde z;y,z)&=& R^l y+\frac{\mu-R^l}{\sigma}\sum_{d=1}^5 z_d -(R^b-R^l)\left(y-\frac{1}\sigma\sum_{d=1}^5 z_d\right)_{-},\\
h^{up}(i,\tilde y,\tilde z;y,z)&=& R^l y+\frac{\mu-R^l}{\sigma}\sum_{d=1}^5 z_d +(R^b-R^l)\left(y-\frac{1}\sigma\sum_{d=1}^5 z_d\right)_{+}.
\end{eqnarray*}
This choice only partially exploits the structure of the generator. It can be applied to any generator which is a linear function of $(y,z)$
 plus a nondecreasing $(R^b-R^l)$-Lipschitz continuous function 
of a linear combination of $(y,z)$. The specific form of the Lipschitz function is not used in this construction of $h^{low}$ and $h^{up}$,
but, of course, the coefficients for the linear combinations must be adjusted to the generator in the obvious way. 
For this semi-generic case the pathwise recursion formulas 
for $\Theta^{h^{up}}$ and $\Theta^{h^{low}}$ can be made explicit in time analogously to the generic case, which was discussed in Example \ref{exmp:generic}.

\begin{table}
\begin{center}
\begin{tabular}{|c|c|c| c| c|}
\hline
Algorithm  \textbackslash $\;\;n$	& 40 & 80 & 120 & 160\\
\hline
fully generic &  $ \twert{ 13.3604  }{ 0.0132    } \; \twert{14.1774   }{ 0.0169     }$ & $ \twert{12.7905    }{ 0.0332    } \; \twert{14.7496   }{0.0407     }$ & $ \twert{12.0148   }{0.0612     } \; \twert{15.8512   }{ 0.0834     }$ & $ \twert{ 10.7872   }{ 0.1005    } \; \twert{ 17.5326    }{ 0.1504      }$  \\
\hline
semi-generic &     $ \twert{ 13.7259   }{ 0.0041    } \; \twert{13.8505   }{0.0046     }$ & $ \twert{  13.6984  }{ 0.0053    } \; \twert{13.8801   }{  0.0059    }$ & 
$ \twert{13.6811    }{  0.0059    } \; \twert{ 13.9136   }{ 0.0071    }$ & $ \twert{ 13.6686     }{  0.0065   } \; \twert{13.9459   }{ 0.0078    }$  \\ 
\hline
\end{tabular}
\end{center}
\caption{Upper and lower price bounds for different time discretizations under the fully generic and semi-generic algorithms. Standard deviations are in brackets.}
\end{table}

Table 3 shows the resulting low-biased and high-biased estimates for the option price $Y_0$ as well as their empirical standard deviations. We observe that the generic bounds 
are not satisfactory in this example. The relative width of the 95\% confidence intervals ranges from about 6.5\% for $n=40$ to more than 65\% for $n=160$ time steps. 
This can be explained by the fact that the approximation of $E_i[\beta_{i+1}Y_{i+1}]$ by $\tilde Z_i$ (which is expressed in terms of just two basis functions) is not yet good enough. 
The quality  of $\tilde Z$ plays an all important role
for the generic bounds due to the appearance
of the terms $\sum_{d=1}^5 |z_d|$ in the definitions of $h^{|low|}$ and $h^{|up|}$. In the semi-generic setting the expressions of the form  $(y-\frac{1}\sigma\sum_{d=1}^5 z_d)_{\pm}$
in $h^{up}$ and $h^{low}$
are much more favorable concerning the approximation error of $E_i[\beta_{i+1}Y_{i+1}]$ by $\tilde Z_i$. Therefore, the semi-generic implementation yields much better 95\% confidence intervals
with a relative width of about 1\% for $n=40$ and still less than 2.5\% for $n=160$ time steps. 

By and large, this example shows that the generic bounds may be too crude, if applied to good but not excellent 
approximations $(\tilde Y,\tilde Z)$, in particular when the $z$-variable of the generator is high-dimensional. Nonetheless very acceptable confidence intervals can still be obtained based 
on the same approximation $(\tilde Y,\tilde Z)$, if some information about the generator is incorporated into the choice of $h^{up}$ and $h^{low}$.

\appendix
\section{Continuous time analogues}

In this appendix we consider BSDEs driven by a Brownian motion $W$ of the form
\begin{equation}\label{cBSDE}
Y_t=\xi+\int_t^T f(s,Y_s,Z_s)ds-\int_t^T Z^\top_s dW_s.
\end{equation}
We assume that the pair $(f,\xi)$ are  standard parameters in the sense of \citet{EKPQ}, p. 18, i.e. square-integrability conditions and a uniform Lipschitz condition on $f$ are in force. Moreover,
$f$ is supposed to be convex in $(y,z)$.

Then, by Proposition 3.4 in \citet{EKPQ}
$$
Y_t=\esssup_{(r,\rho)\in \mathcal{U}^2_t(f^\#)} E\left[\left. \gamma_{t,T}(r,\rho)\xi -\int_t^T \gamma_{t,s}(r,\rho) f^\#(s,r_s,\rho_s)ds\right|\mathcal{F}^W_t\right],
$$
 where $(\mathcal{F}^W_t)_{t\in [0,T]}$ is the augmented filtration generated by the driving Brownian motion,
$$
\gamma_{t,s}=\exp\left\{\int_t^s r_u du +\int_t^s \rho_u^\top dW_u \right\},
$$
and the supremum runs over the set
$$
\mathcal{U}^2_t(f^\#):=\left\{(r_s,\rho_s)_{s\geq t} \textnormal{ predictable};\; \int_t^T E[|f^\#(s,r_s,\rho_s)|^2]ds<\infty \right\}.
$$
This is the non-reflected continuous time analogue to the primal optimization problem in Theorem \ref{thm:primal} in a Brownian environment.

We now derive a continuous time version of 
the pathwise approach to the dual minimization problem in Theorem \ref{thm:convexdual}. On the one hand this continuous time version sheds additional light on the need to use a 
$(1+D)$-dimensional martingale in the upper bound construction in discrete time. On the other hand it might serve as a starting point for the design of alternative upper bound algorithms.

We shall make use of some basic tools from Malliavin calculus. 
For the corresponding definitions and notations we refer to \citet{Nu}. In order to simplify the notation, we assume that the driving Brownian motion is one-dimensional.
 Given a stochastic process
$\theta$ such that $\theta_t$ is Malliavin differentiable for a.e. $t\in[0,T]$, we denote by $D \theta$ the Malliavin derivative of $\theta$.
Notice that  the field $(D_s\theta_t)_{s,t \in [0,T]^2}$ is only defined almost everywhere on $[0,T]^2$, and consequently
the trace  $D_t\theta_t$ of $D_s\theta_t$ is not well-defined. We shall therefore make use of the one-sided trace $(D^+\theta)_t$, as introduced on p. 173 in \citet{Nu}
for $p=2$.

Now given a martingale $M^0$ such that $M^0_T\in \mathbb{D}^{1,2}$, (i.e. the random variable $M^0_T$ is Malliavin differentiable with square-integrable Malliavin derivative), 
we say that a possibly non-adapted process $\theta$ is a $M^0$-solution of
\begin{equation}\label{antBSDE}
-d\theta_t=f(t,\theta_t,(D^+ \theta)_t)dt-dM^0_t,\quad \theta_T=\xi
\end{equation}
if $E[\int_0^T |\theta_t|^2 dt  ] <\infty$, $(D^+\theta)_t$ exists, $f(\cdot,\theta_\cdot, (D^+\theta)_\cdot) \in \mathbb{L}^{1,2}$,
and for every $t\in[0,T]$
$$
\theta_t=\xi+\int_t^T f(s,\theta_s,(D^+ \theta)_s)ds-(M_T^0-M^0_t).
$$
Now suppose that $\theta$ is a $M^0$-solution for some martingale $M^0$ such that $M^0_T\in \mathbb{D}^{1,2}$. Define $\tilde Y_t=E[\theta_t|\mathcal{F}^W_t]$, $\tilde Z_t=E[(D^+ \theta)_t|\mathcal{F}^W_t]$,
and 
$$
c_s= E[f(s,\theta_s,(D^+ \theta)_s)|\mathcal{F}^W_s]- f(s,E[\theta_s|\mathcal{F}^W_s],E[(D^+ \theta)_s|\mathcal{F}^W_s]).
$$
Then,
\begin{equation}\label{A1}
\tilde Y_t+\int_0^t \left(f(s,\tilde Y_s, \tilde Z_s)+c_s\right)ds=E\left[\left. \xi+\int_0^T E[f(s,\theta_s,(D^+\theta)_s)|\mathcal{F}^W_s] ds \right|\mathcal{F}^W_t\right]=:\tilde M_t.
\end{equation}
Assuming that $\xi\in \mathbb{D}^{1,2}$, we next note that
\begin{equation}\label{A2}
 \tilde Z_t=E\left[\left.D_t\xi +\int_t^T D_tf(s,\theta_s, (D^+ \theta)_s)ds\right|\mathcal{F}^W_t\right].
\end{equation}
Indeed, by the martingale representation theorem and Lemma 1.3.4 in \citet{Nu}, there is an adapted process $u\in \mathbb{L}^{1,2}$ such that 
$$
M^0_t=M^0_0+\int_0^t u_s dW_s.
$$
Then, by Proposition 1.3.8 and the same argument as in Proposition 3.1.1 in \citet{Nu},
$$
(D^+ \theta)_t=D_t\xi +\int_t^T D_tf(s,\theta_s, (D^+ \theta)_s)ds - \int_t^T D_tu_s dW_s.
$$
The last integral is a martingale increment by adaptedness and square-integrability of the integrand. Hence, taking conditional expectation yields (\ref{A2}). 
We are now in the position to link $\tilde Z$ to the martingale $\tilde M$, which was defined in (\ref{A1}).
By the Clark-Ocone formula \citep[][Proposition 1.3.14]{Nu}, we obtain
\begin{eqnarray*}
 \tilde M_t-\tilde Y_0&=&\int_0^t E\left[\left. D_r \left(\xi+\int_0^T E[f(s,\theta_s,(D^+ \theta)_s)|\mathcal{F}^W_s] ds\right)\right|\mathcal{F}^W_r\right] dW_r \\
&=&  \int_0^t E\left[\left.D_r\xi + \int_r^T E[D_r f(s,\theta_s,(D^+ \theta)_s)|\mathcal{F}^W_s] ds\right|\mathcal{F}^W_r\right] dW_r
\\
&=&  \int_0^t E\left[\left.D_r\xi + \int_r^T D_r f(s,\theta_s,(D^+ \theta)_s)\right|\mathcal{F}^W_r\right] dW_r \\
&=&  \int_0^t \tilde Z_r dW_r,
\end{eqnarray*}
where we used Proposition 1.2.8 from \citet{Nu} to interchange Malliavin derivative and conditional expectation, and (\ref{A2}). 
Since $\tilde Y_T=E[\theta_T|\mathcal{F}^W_T]=\xi$, we conclude, thanks to (\ref{A1}), that $(\tilde Y, \tilde Z)$ solves the BSDE
$$
\tilde Y_t=\xi+\int_t^T  (f(s,\tilde Y_s, \tilde Z_s)+c_s)ds - \int_t^T \tilde Z_s dW_s.
$$
By the convexity of $f$  we observe that $c_s\geq 0$. Hence, by the comparison theorem \citep[see][Theorem 2.2]{EKPQ}, we end up with
$$
E[\theta_t|\mathcal{F}^W_t]=\tilde Y_t \geq Y_t.
$$
Finally, Proposition 5.3 in \citet{EKPQ} shows that the unique adapted solution $(Y,Z)$ to BSDE (\ref{cBSDE}) satisfies $Z_t=(D^+Y)_t$ under some technical conditions on $f$ and $\xi$,
which we assume
from now on. In particular, $Y$ is a $M^0$-solution to (\ref{antBSDE}) for $M^0=\int_0^\cdot Z_s dW_s$. 
Summarizing the above, we arrive at the following result:
\begin{prop}
Suppose that the assumptions of Proposition 5.3 in \citet{EKPQ} on $(f,\xi)$ are in force. Then,
$$
Y_t=\essinf_\theta  E[\theta_t|\mathcal{F}^W_t],
$$
where the infimum runs over the set of those processes $\theta$, which are $M^0$-solutions of (\ref{antBSDE}) for some martingale $M^0$ such that $M^0_T\in \mathbb{D}^{1,2}$.
\end{prop}
Comparing this result with the discrete time result in Theorem \ref{thm:convexdual}, we immediately observe a major difference: In continuous time only the choice of a one-dimensional martingale $M^0$
is required, while in
discrete time one additionally needs to choose a $D$-dimensional martingale $M$. This phenomenon is easily explained. Notice first that, under at most technical conditions,
\begin{eqnarray*}
(D^+\theta)_t =\lim_{\epsilon\downarrow 0} \frac{1}{\epsilon} \int_t^{t+\epsilon} D_s \theta_{t+\epsilon}ds=
\lim_{\epsilon\downarrow 0}  \left(\frac{W_{t+\epsilon}-W_t}{\epsilon} \theta_{t+\epsilon} -
\frac{W_{t+\epsilon}-W_t}{\epsilon} \diamond  \theta_{t+\epsilon} \right),
\end{eqnarray*}
 where the diamond denotes the Wick product,
see Theorem 6.8 in \citet{NOP}. The first term on the right hand side corresponds to the expression $\beta_{i+1}\theta^{up}_{i+1}$ in (\ref{pathwise}), when
$\beta_{i+1}(t_{i+1}-t_i)$ equals the truncated Brownian increment over $[t_i,t_{i+1}]$. The second term on the right hand side has zero conditional expectation, because the Wick product interchanges with the conditional expectation, i.e.
$$
E\left[\left. \frac{W_{t+\epsilon}-W_t}{\epsilon} \diamond  \theta_{t+\epsilon} \right|\mathcal{F}^W_t\right]=E\left[\left. \frac{W_{t+\epsilon}-W_t}{\epsilon} \right|\mathcal{F}^W_t\right]\diamond E[\theta_{t+\epsilon} |\mathcal{F}^W_t]=0,
$$
see e.g. Lemma 6.20 in \citet{NOP}.
As one cannot expect that the Wick product $\beta_{i+1}\diamond \theta_{i+1}$ can be computed in closed form, a generic term with zero conditional expectation, namely the martingale increment
$M_{i+1}-M_i$, is subtracted
in (\ref{pathwise}). Due to the convexity of $f$, subtracting this generic term with zero conditional expectation pushes the solution of the recursion (\ref{pathwise}) upwards.

\end{document}